\documentclass[lettersize,journal]{IEEEtran}
\usepackage{booktabs}
\usepackage{graphics,colortbl,multicol,lipsum,xcolor,graphicx,color}
\usepackage[cmex10]{amsmath}
\usepackage[colorlinks=true,bookmarks=false,citecolor=blue,urlcolor=blue]{hyperref}
\usepackage{amsthm,mathrsfs,mathtools,amsbsy,amsfonts,amssymb,}
\usepackage{bbm}
\usepackage{setspace,float,pstool,lettrine,newclude}
\usepackage[normalem]{ulem}
\usepackage{enumitem}
\usepackage{latexsym,algpseudocode,gensymb,balance,multirow,bm,wasysym}
\usepackage{hhline}
\usepackage{cite}
\usepackage[linesnumbered,ruled,vlined]{algorithm2e}
\SetKwInput{KwInput}{Input}  
\SetKwInput{KwOutput}{Output}
\SetKw{KwBy}{by}
\makeatletter
\newcommand{\nosemic}{\renewcommand{\@endalgocfline}{\relax}}
\newcommand{\dosemic}{\renewcommand{\@endalgocfline}{\algocf@endline}}
\let\oldnl\nl
\newcommand{\nonl}{\renewcommand{\nl}{\let\nl\oldnl}}
\makeatother

\ifCLASSOPTIONcompsoc
\usepackage[caption=false,font=normalsize,labelfont=sf,textfont=sf]{subfig}
\else
\usepackage[caption=false,font=footnotesize]{subfig}
\fi
\ifCLASSOPTIONcompsoc
\usepackage[caption=false,font=normalsize,labelfon
t=sf,textfont=sf]{subfig}
\else
\usepackage[caption=false,font=footnotesize]{subfig}
\fi

\DeclarePairedDelimiterX\MeijerM[3]{\lparen\!}{\rparen}%
{\,#3\delimsize\vert\begin{smallmatrix}#1 \\ #2\end{smallmatrix}}

\newcommand\MeijerG[8][]{%
  G^{\,#2,#3}_{#4,#5}\MeijerM[#1]{#6}{#7}{#8}}

\WithSuffix\newcommand\MeijerG*[7]{%
  G^{\,#1,#2}_{#3,#4}\MeijerM*{#5}{#6}{#7}}

\allowdisplaybreaks

\newtheorem{proposition}{Proposition}

\graphicspath{{figures/}}
\allowdisplaybreaks

\newcommand{\RNum}[1]{\uppercase\expandafter{\romannumeral #1\relax}}

\usepackage{mathtools}

\usepackage{accents}

\begin{document}

\title{Integrated Push-and-Pull Update Model for Goal-Oriented Effective Communication}
\author{Pouya Agheli, 
\IEEEmembership{Graduate Student Member, IEEE}, Nikolaos Pappas, 
\IEEEmembership{Senior Member, IEEE}, \\
Petar Popovski, \IEEEmembership{Fellow, IEEE},
and Marios Kountouris, \IEEEmembership{Fellow, IEEE}.
\thanks{P. Agheli and M. Kountouris are with the Communication Systems Dept., EURECOM, France, email: \texttt{agheli@eurecom.fr}. M. Kountouris is also with the Dept. of Computer Science and Artificial Intelligence, University of Granada, Spain, email: \texttt{mariosk@ugr.es}. N. Pappas is with the Dept. of Computer and Information Science, Linköping University, Sweden, email: \texttt{nikolaos.pappas@liu.se}. P. Popovski is with the Dept. of Electronic Systems, Aalborg University, Denmark, email: \texttt{petarp@es.aau.dk}. Part of this work is presented in \cite{agheli2023effective}. 
}}

\maketitle

\begin{abstract}
This paper studies decision-making for goal-oriented effective communication. We consider an end-to-end status update system where a sensing agent (SA) observes a source, generates and transmits updates to an actuation agent (AA), while the AA takes actions to accomplish a goal at the endpoint. We integrate the push- and pull-based update communication models to obtain a \emph{push-and-pull} model, which allows the transmission controller at the SA to decide to push an update to the AA and the query controller at the AA to pull updates by raising queries at specific time instances. To gauge effectiveness, we utilize a \emph{grade of effectiveness} (GoE) metric incorporating updates' freshness, usefulness, and timeliness of actions as qualitative attributes. We then derive effect-aware policies to maximize the expected discounted sum of updates' effectiveness subject to induced costs. The effect-aware policy at the SA considers the potential effectiveness of communicated updates at the endpoint, while at the AA, it accounts for the probabilistic evolution of the source and importance of generated updates. Our results show the proposed push-and-pull model outperforms models solely based on push- or pull-based updates both in terms of efficiency and effectiveness. Additionally, using effect-aware policies at both agents enhances effectiveness compared to periodic and/or probabilistic effect-agnostic policies at either or both agents.
\end{abstract}

\IEEEpeerreviewmaketitle

\section{Introduction}
The emergence of cyber-physical systems empowered with interactive and networked sensing and actuation/monitoring agents has caused a shift in focus from extreme to sustainable performance. Emerging networks aim to enhance effectiveness in the system while substantially improving resource utilization, energy consumption, and computational efficiency. The key is to strive for a minimalist design, frugal in resources, which can scale effectively rather than causing network over-provisioning. This design philosophy has crystallized into the goal-oriented and/or semantic communication paradigm, holding the potential to enhance the efficiency of diverse network processes through a parsimonious usage of communication and computation resources \cite{kountouris2021semantics,popovski2020semantic}. Under an effectiveness perspective, a message is generated and conveyed by a sender \emph{if} it has the potential to have the \emph{desirable effect} or the \emph{right impact} at the destination, e.g., executing a critical action, for accomplishing a specific goal. This promotes system scalability and efficient resource usage by avoiding the acquisition, processing, and transportation of information turning out to be ineffective, irrelevant, or useless.

Messages, e.g., in the form of status update packets, are communicated over existing networked intelligent systems mostly using a \emph{push-based} communication model. Therein, packets arrived at the source are sent to the destination based on decisions made by the source, regardless of whether or not the endpoint has requested or plans to utilize these updates to accomplish a goal. In contrast, in a \emph{pull-based} model, the endpoint decides to trigger and requests packet transmissions from the source and controls the time and the type of generated updates \cite{pullB,pullC, pullD, pullDD, pullE, pullF, hatami2021aoi}. Nevertheless, this model does not consider the availability of the source to generate updates or the usefulness of those updates. To overcome these limitations, we propose an \emph{integrated push-and-pull} model that involves both agents/sides in the decision-making process, thereby combining push- and pull-based paradigms in a way that mitigates their drawbacks. In either model, decisions at the source or endpoint could influence the effectiveness of communicated updates. Therefore, we can categorize decision policies into \emph{effect-aware} and \emph{effect-agnostic}. Under an effect-aware policy, the source adapts its decisions by taking into account the effects of its communicated packets at the endpoint. Likewise, the endpoint initiates queries based on the evolution of the source and the expected importance of pulled updates. Under the effect-agnostic policy, however, decisions are made regardless of their consequent effect on the performance.


We illustrate the motivation behind the effect-aware integrated model with a real-world example: \textit{Consider a remote patient monitoring system. A medical expert (the endpoint) is responsible for responding to a patient's vital signs (the source) but may not always be available due to other responsibilities. The push-based model risks overwhelming the expert with irrelevant or redundant information, or worse, causing critical updates to be missed amid competing demands. Conversely, the pull-based model requires the expert to check the patient’s status at predefined intervals, but this rigidity may lead to missing urgent or critical changes in the patient's condition. The push-and-pull model strikes a balance between these approaches by integrating the strengths of both. It allows the experts to receive timely updates and act while maintaining the flexibility to manage other responsibilities or to respond effectively without dedicating all their time to a single patient. By integrating effect-aware decision-making, updates are prioritized based on their importance and the expert's availability to act. Furthermore, by having access to and continually monitoring the patient’s evolving health condition, the expert can dynamically adjust query timing and align decisions with the expected relevance of updates.}

In this work, we investigate a time-slotted end-to-end status update system where a sensing agent observes an information source and communicates updates/observations in the form of packets with an actuation agent. The actuation agent then takes action based on the updates successfully received as a means to accomplish a subscribed goal at the endpoint. We develop an integrated push-and-pull model, which allows both agents to make decisions based on their local policies or objectives. In particular, a transmission controller at the sensing agent decides to either send or drop update packets according to their potential usefulness at the endpoint. On the other side, a query controller at the actuation agent also determines the time instances around which the actuator should perform actions in the form of raising queries. In that sense, effective updates are those that result in the right impact and actuation at the endpoint. Those queries, however, are \emph{not} communicated to the sensing agent. Instead, the actuation agent acknowledges the sensing agent of effective updates. With prior knowledge that the effectiveness of updates depends on the actuator's availability to perform actions, the sensing agent can infer the initiated queries using those acknowledgments. A time diagram showing processes at both agents is depicted in Fig.~\ref{intro:fig}.

We introduce a metric to measure the effectiveness and the significance of updates and derive a class of optimal policies for each agent that makes effect-aware decisions to maximize the long-term expected effectiveness of update packets communicated to fulfill the goal subject to induced costs. These costs may arise from transmitting updates, initiating queries, and the availability of the actuation agent to respond and take action, all of which depend on the communication model. To do so, the agent first needs to estimate the necessary system parameters for making the right decisions. Our analytical and simulation results show that the integrated push-and-pull model comes with a higher energy efficiency compared to the push-based model and better effectiveness performance compared to the pull-based one. Moreover, we observe that utilizing effect-aware policies at both agents significantly improves the effectiveness performance of the system in the majority of the cases, with a large gap compared to periodic and probabilistic effect-agnostic policies at either or both agents. Accordingly, we demonstrate that the solution to find an optimal effect-aware policy at each agent converges to a \emph{threshold-based} agent decision framework where the agent can timely decide based on an individual lookup map in hand and threshold boundaries computed to satisfy the goal.
\begin{figure}[t!]
    \centering
    \includegraphics[width=1\linewidth]{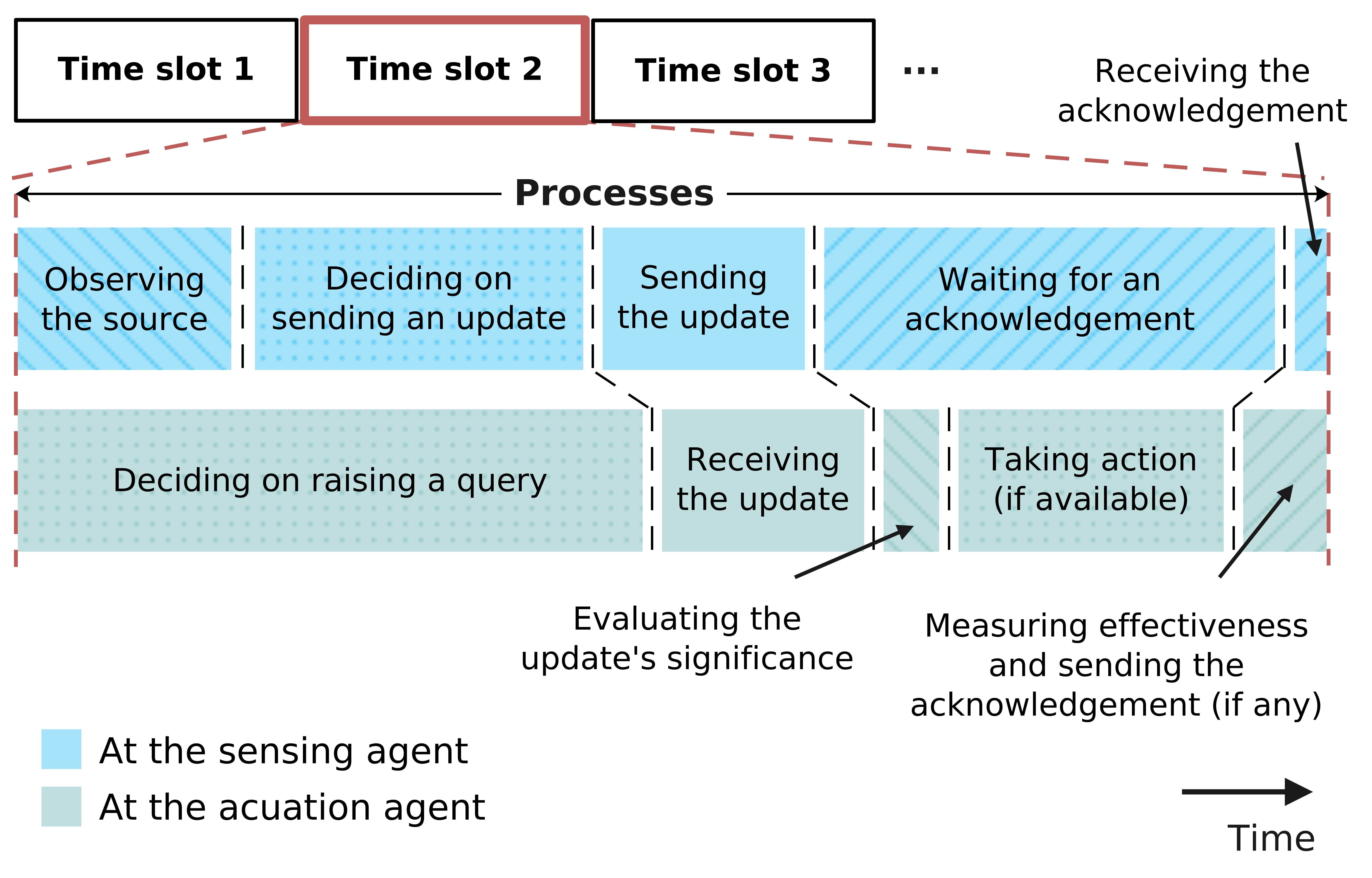}
    \vspace{-0.6cm}
    \caption{A timing diagram of processes involving the sensing and actuation agents, illustrating interactions and update communications leading to actions.}
    \label{intro:fig}
\end{figure}

\subsection{Related Works}
This paper widely broadens prior work on push-based and (query-) pull-based communications by enabling both agents to make decisions so as to maximize the effectiveness of communicated updates in the system. The pull-based communication model has been widely analyzed but not limited to \cite{pullB, pullC, pullD, pullDD, pullE, pullF, hatami2021aoi}. In \cite{pullB}, a metric called effective age of information (EAoI), which comprises the effects of queries and the freshness of updates in the form of the age of information (AoI) \cite{NowAoI,yates2021age, pappas2023age}, is introduced. Query AoI (QAoI), which is similar to the EAoI, is utilized in \cite{pullC, pullD, pullDD, pullE, pullF}. Following the same concept as the QAoI, on-demand AoI is introduced in \cite{hatami2021aoi,hatami2022demand}. Probe (query)-based active fault detection where actuation or monitoring agents adaptively decide to probe sensing agents to detect probable faults at the endpoint is studied in \cite{stamatakis2022semantics}. Most prior work has employed a pull-based communication model and focuses on the freshness and timeliness of information. In this work, we consider multiple information attributes and propose a \emph{grade of effectiveness} metric to measure the effectiveness of updates, which goes beyond metrics, including AoI, EAoI, QAoI, on-demand AoI, age of incorrect information (AoII) \cite{maatouk2022age}, and value of information (VoI) \cite{VoI_USSR, VoI}. In particular, we focus on the freshness of successfully received updates and the timeliness of performed actions as two attributes of interest through the link level, as well as on the usefulness/significance (semantics) of the updates to fulfill the goal at the source level. 

This paper extends our prior work \cite{agheli2023effective}, which only considers a pull-based model and an effectiveness metric with two freshness and usefulness attributes. As such, \cite{agheli2023effective} conveys a special form of the decision problem we solve here. In this work, we generalize the problem to a push-and-pull communication model, considering that the sensing and actuation agents individually make decisions and converge to a point where they can transmit updates and initiate queries, respectively, which maximize the effectiveness of updates and result in the right impact at the endpoint. Importantly, we assume that the source distribution is not known to the actuation agent, and the sensing agent does not have perfect knowledge of the goal. Therefore, the agents must estimate their required parameters separately. This approach is substantially different from the one in our previous work and other state-of-the-art approaches.

\subsection{Contributions}
The main contributions can be briefly outlined as follows.
\begin{itemize}
    \item We develop an \emph{integrated push-and-pull update communication model} owing to which both agents have decision-making roles in the acquisition and transmitting of updates and take appropriate actions to satisfy the goal, following the paradigm of goal-oriented communications. With this, the system becomes adaptable from an effectiveness viewpoint compared to the conventional push- and pull-based models.

    \item We use a \emph{grade of effectiveness} metric to capture the timely impact of communicated updates at the endpoint, which relies on the freshness of successfully communicated updates, the timeliness of actions performed, and the usefulness of those updates in fulfilling the goal. Our approach maps multiple information attributes into a unique metric that measures the impact or effect each status update packet traveling over the network can offer. 

    \item We obtain optimal model-based control policies for agents that make effect-aware decisions to maximize the discounted sum of updates' effectiveness while keeping the induced costs within certain constraints. To achieve this, we formulate an optimization problem, derive its dual form, and propose an iterative algorithm based on dynamic programming to solve the decision problem separately from each agent's perspective.

    \item We demonstrate that the integrated push-and-pull model offers higher energy efficiency than the push-based model and better effectiveness performance compared to the pull-based one. We also show that applying effect-aware policies at both agents results in better performance than in the scenarios where one or both agents utilize effect-agnostic policies. We also broaden our results via deriving model-free decisions using \emph{reinforcement learning}. Eventually, we provide a lookup map presenting optimal decisions for each agent that applies the effect-aware policy based on the given solution. This allows the agent to make decisions on time by merely looking up the map with the obtained threshold-based policy for the goal.
\end{itemize}

\textit{Notations:} $\mathbb{R}$, $\mathbb{R}_0^+$, and $\mathbb{N}$ indicate the sets of real, non-negative real, and natural numbers, respectively. $\mathbb{E}[\cdot]$ denotes the expectation operator, $\lvert\cdot \rvert$ depicts the absolute value operator, $\mathbbm{1}\{\cdot\}$ is the indicator function, $\mathcal{O}(\cdot)$ shows growth rate of a function, and $\lceil\cdot\rceil$ is the ceiling function.

\section{System Model}\label{sec2}
We consider an end-to-end communication system in which a \emph{sensing agent} (SA) sends messages in a time-slotted manner to an \emph{actuation agent} (AA) as a means to take effective action at the endpoint and satisfy a \emph{subscribed goal}. 
Specifically, the SA observes a source and generates \emph{status update packets} in each time slot, and a transmission controller decides whether to transmit that observation or not, following a specific policy. We assume that the source has finite-dimensional realizations and that observation at the $n$-th, $\forall n \in \mathbb{N}$, time slot is assigned a rank of importance $v_n$ from a finite set $\mathcal{V}=\{\nu_i\,\lvert\,i\in\mathcal{I}\}$, with $\mathcal{I}=\{1, 2, \dots, \lvert\mathcal{V}\rvert\}$, based on its significance or usefulness for satisfying the goal, measured or judged at the source level.\footnote{To determine the usefulness of an update, we can use the same \emph{metavalue} approach proposed in \cite[Section~III-A]{agheli2023semantic}.} The elements of $\mathcal{V}$ are independent and identically distributed (i.i.d.) with probability $p_i = p_\nu(\nu_i)$ for the $i$-th outcome, where $p_\nu(\cdot)$ denotes a given probability mass function (pmf).\footnote{A more elaborated model could consider the importance of a realization dependent on the most recently generated update at the SA. This implies that a less important update increases the likelihood of a more significant update occurring later, which can be captured utilizing a learning algorithm.}

The AA is assisted by a query controller that decides to initiate queries and pull new updates according to a certain policy. A received packet at the AA has a satisfactory or sufficient impact at the endpoint if that update achieves a minimum effectiveness level subject to the latest query initiated and the AA's availability to act on it. An effective update communication is followed by an \emph{acknowledgment of effectiveness} (E-ACK) signal sent from the AA to the SA to inform about the effective update communication. We assume all transmissions and E-ACK feedback occur over packet erasure channels (PECs), with $p_{\epsilon}$ and $p_{\epsilon}^\prime$ being the erasure probabilities in the forward communication and the acknowledgment links, respectively. Therefore, an E-ACK is not received at the SA due to either \emph{ineffective} update communication or erasure in the acknowledgment (backward) channel. With this interpretation, channel errors lead to graceful degradation of the proposed scheme. An initiated query does \emph{not} necessarily need to be shared with the SA. As discussed in Section~\ref{sec3c}, the SA can deduce the initiation of a query or the availability of the AA to take action from a successful E-ACK, given prior knowledge that an update can be effective \emph{only if} it arrives within the period during which the AA is available to act.

In this model, we consider the goal to be subscribed at the endpoint, with the AA fully aware of it. On the other hand, the SA does not initially know the goal but learns which updates could be useful to accomplish the goal based on the received E-ACK and observations' significance. Meanwhile, the AA is not aware of the evolution of the source or the likely importance of observations, attempting to approximate it from arrival updates. Consequently, the agents might use different bases to measure the usefulness of the updates and may need to adjust their criteria or valuation frameworks to account for possible changes in goals over time. Finally, we assume that update acquisition, potential communication, and waiting time for receiving an E-ACK occur within \emph{one} slot.

\subsection{Communication Model}\label{sec2a}
The following three strategies can be employed for effective communication of status updates.
\subsubsection{Push-based} Under this model, the SA pushes its updates to the AA, taken for instance based on the source evolution, without considering whether the AA has requested them or is available to take any action upon receipt. This bypasses the query controller, enabling the SA to directly influence actions at the AA side.
\subsubsection{Pull-based} In this model, the query controller plays a central role in the generation of update arrivals at the AA by pulling those updates from the SA. Here, the AA can \emph{only} take action when queries are initiated. However, this model excludes the SA from generating and sending updates.
\subsubsection{Push-and-pull} This model arises from integrating the push- and pull-based models so that the transmission and query controllers individually decide to transmit updates and send queries, respectively. Thereby, the AA is provided with a level of flexibility where it is also able to take some actions \emph{beyond} query instances within a \emph{limited} time. As a result, the effectiveness of an update packet depends on both agents' decisions. Dismissing the decision of either agent transforms the push-and-pull model into the push- or pull-based model.

\subsection{Agent Decision Policies}\label{sec2b}
We propose that the agents can adhere to the following decision policies, namely \emph{effect-agnostic} and \emph{effect-aware}, for transmitting updates or raising queries to satisfy the goal.

\subsubsection{Effect-agnostic} This policy employs a predetermined schedule or random process (e.g., Poisson, binomial, or Markov \cite{pullC, pullD, pullDD, pullE, pullF}) to send updates (initiate queries) from (by) the SA (AA), without considering their impact at the endpoint. We define a \emph{controlled update transmission (query) rate} specifying the expected constant number of updates (queries) to be communicated (initiated) within a period. Also, as the effect-agnostic policy does not account for what might be happening in the other agent during the decision time, there exists an \emph{aleatoric} uncertainty associated with random updates (queries).

\subsubsection{Effect-aware} The effect-aware policy takes into consideration the impacts of both agents' decisions at the endpoint. In this regard, the SA (AA) predicts the effectiveness status at the endpoint offered by a sent update that is potentially received at the AA (the usefulness of a possible update at the source). Then, based on this prediction, the agent attempts to adapt transmission (query) instants and send (pull) updates in the right slots. This policy comes with an \emph{epistemic} uncertainty because decisions are made according to probabilistic estimations, not accurate knowledge. However, such uncertainty can be decreased using learning or prediction techniques.

\section{Effectiveness Analysis Metrics}\label{sec3}
To achieve the right effect at the endpoint, an update packet that is successfully received at the AA has to satisfy a set of qualitative attributes, captured by the metrics as follows.

\subsection{Grade of Effectiveness Metric}\label{sec3a}
We introduce a \emph{grade of effectiveness} (GoE) metric that comprises several qualitative attributes and characterizes the amount of impact an update makes at the endpoint. Mathematically speaking, the GoE metric is modeled via a composite function $\operatorname{GoE}_n = (f \circ g)(\mathcal{I}_n)$ for the $n$-th time slot. Here, $g:\mathbb{R}^x \to \mathbb{R}^y, x\geq y$, is a (nonlinear) function of $x \in \mathbb{N}$ information attributes $\mathcal{I}_n \in \mathbb{R}^x$, and $f: \mathbb{R}^y \to \mathbb{R}$ is a context-aware function.\footnote{The GoE metric in this form can be seen as a special case of the semantics of information (SoI) metric introduced in \cite{kountouris2021semantics,pappas2021goal,agheli2023semantic}.} The particular forms of functions $f$ and $g$ could vary according to different subscribed goals and their relevant requirements. 

In this paper, without loss of generality, we consider \emph{freshness} of updates and \emph{timeliness} of actions as the main contextual attributes. The first comes in the form of the AoI metric, which is denoted by $\Delta_n\in\mathbb{N}$. The second is measured from the action's \emph{lateness}, denoted by $\Theta_n\in\mathbb{N}$. Thereby, we can formulate the GoE metric as follows
\begin{align}\label{sec3:eq1}
    \operatorname{GoE}_n = f_g\Big(g_\Delta(\hat{v}_n, \Delta_n), g_\Theta(\Theta_n); g_c(C_n)\Big)
\end{align}
where $C_n\in\mathbb{R}_0^+$ represents the overall cost incurred at the $n$-th time slot. Also, $g_\Delta : \hat{\mathcal{V}}\times \mathbb{N} \rightarrow \mathbb{R}_0^+$, $g_\Theta : \mathbb{N} \rightarrow \mathbb{R}_0^+$, and $g_c : \mathbb{R}_0^+ \rightarrow \mathbb{R}_0^+$ are \emph{penalty} functions, and $f_g : {\mathbb{R}_0^+}\times{\mathbb{R}_0^+}\times{\mathbb{R}_0^+} \rightarrow \mathbb{R}_0^+$ is a non-decreasing \emph{utility} function. Moreover, $g_\Delta$,  $g_\Theta$, and $g_c$ are non-increasing with respect to (w.r.t.) $\Delta_n$, $\Theta_n$, and $C_n$, respectively, while $g_\Delta$ is non-decreasing w.r.t. $\hat{v}_n$. Here, $\hat{v}_n$ is the usefulness of the received update from the endpoint's viewpoint at the $n$-th slot. Thus, we assume that $\hat{v}_n$ belongs to the set $\hat{\mathcal{V}}=\{0\}\cup  \{\hat{\nu}_j\,\lvert\,\hat{\nu}_j>0, j\in\mathcal{J}\}$ with i.i.d. elements, where $\mathcal{J}=\{1, 2, \dots, \lvert\hat{\mathcal{V}}\rvert-1\}$, the $j$-th element has probability $q_j = p_{\hat{\nu}}(\hat{\nu}_j)$, and $p_{\hat{\nu}}(\cdot)$ is a pmf derived in Section~\ref{sec5a}. Since the packet is sent over a PEC, $\hat{v}_n=0$ if it is erased or the update ends up being useless at the endpoint.

\subsubsection{AoI} Measuring the freshness of correctly received updates at the AA within a query slot, the AoI is defined as $\Delta_n = n - u(n)$, where $\Delta_0 = 1$ and $u(n)$ is the slot index of the latest successful update, which is given by 
\begin{align}\label{sec3:eq2}
u(n) = \operatorname{max}\big\{m\,\lvert\,m \leq n, \beta_m(1-\epsilon_m)=1\big\} 
\end{align}
with $\epsilon_m \in \{0, 1\}$ being the channel erasure at the $m$-th slot. In addition, $\beta_m\in\{0, 1\}$ indicates the query controller’s decision, where $\beta_m=1$ means pulling the update; otherwise, $\beta_m=0$.
    
\subsubsection{Action lateness} The lateness of an action performed at the $n$-th time slot in relevance to a query initiated at the $n^\prime$-th slot, $n^\prime\leq n$, is calculated as follows
\begin{align}\label{sec3:eq3}
    \Theta_n = (1\!-\!\beta_n)(n\!-\!n^\prime),
\end{align}
which is valid for $\Theta_n < \Theta_{\rm max}$. Herein, $\Theta_{\rm max}$ shows the \emph{width of action window} within which the AA can act on each query based on update arrivals from the SA. Outside the dedicated action window for the SA, the AA might undertake other tasks or communicate with other agents. Employing the push-based, pull-based, and push-and-pull update communication models, we have $\Theta_{\rm max} =\infty$, $\Theta_{\rm max} = 1$, and $\Theta_{\rm max} >1$, respectively. Fig.~\ref{Sec2:fig2} shows the action and idle windows for different models. A wider action window enables higher flexibility during heavy action loads at the cost of longer actuation availability.
\begin{figure}[t!]
    \centering
    \includegraphics[width=1\linewidth]{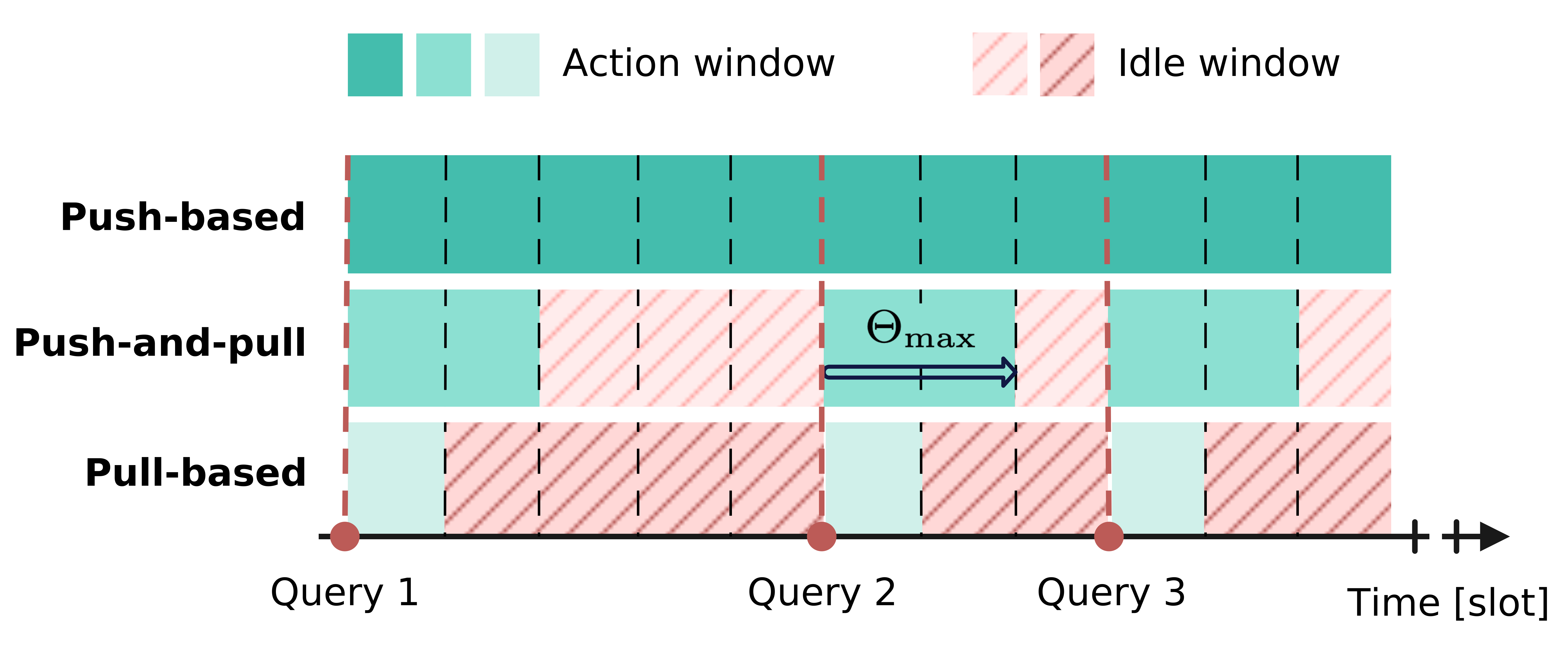}
    \vspace{-0.5cm}
    \caption{The outline of the action and idle windows in different models.}
    \label{Sec2:fig2}
\end{figure}

\subsection{Special Forms of the GoE}
The GoE metric’s formulation in \eqref{sec3:eq1} can simply turn into the QAoI and the VoI metrics as special cases. In this regard, we obtain a penalty function of the QAoI such that $\operatorname{GoE}_n = g_\Delta(\Delta_n)$ if we set $\Theta_{\rm max} = 1$, assume linear $g_\Theta(\cdot)$, and overlook updates' usefulness and cost. In addition, by removing the concepts of query and time, hence the freshness and timeliness in the GoE’s definition, we arrive at a utility function of the VoI, i.e., $\operatorname{GoE}_n = f_g(\hat{v}_n;g_c(C_n))$.

\subsection{Effectiveness Indicator}\label{sec3c}
An update at the $n$-th time slot is considered effective at the system level if its $\operatorname{GoE}_n$ is higher than a \emph{target effectiveness grade}, which is called $\operatorname{GoE}_{\rm tgt}$ and is necessary to satisfy the goal. Let us define $E_n$ as an \emph{effectiveness indicator} at the $n$-th time slot. Thus, we can write
\begin{align}\label{sec3:eq4}
    E_n = \mathbbm{1}\{\operatorname{GoE}_n \geq \operatorname{GoE}_{\rm tgt} \land\, \Theta_n < \Theta_{\rm max}\}.
\end{align}
The second condition in \eqref{sec3:eq4} appears from \eqref{sec3:eq3}. According to \eqref{sec3:eq4}, an update could be effective \emph{only if} it arrives within the action window of the AA. Hence, a consequent E-ACK shared with the SA can imply the initiation of a query or the availability of the AA to take action. Given the values of $\Delta_n$ and $\Theta_n$, and by inserting \eqref{sec3:eq1} into \eqref{sec3:eq4}, we reach a \emph{target usefulness} level $v_{\rm tgt}$ as the importance threshold that the update should exceed to be considered effective. In this case, if $\Theta_n < \Theta_{\rm max}$, we have
\begin{align}\label{sec3:eq5}
    &v_{\rm tgt} = \nonumber \\
    &\operatorname{min} \left\{ \hat{\nu}_j\,\lvert\, \hat{\nu}_j\in\hat{\mathcal{V}}, f_g(g_\Delta(\hat{\nu}_j, \Delta_n), g_\Theta(\Theta_n)) \geq \operatorname{GoE}_{\rm tgt}\right\}\!;
\end{align} 
otherwise, $v_{\rm tgt} = \operatorname{max}\{\hat{\nu}_j\,\lvert\,\hat{\nu}_j\in\hat{\mathcal{V}}\}$. In \eqref{sec3:eq5}, $v_{\rm tgt}$ can be computed by exhaustive search.

\section{Model-Based Agent Decisions}\label{sec4}
Here, we first formulate a decision problem for effect-aware policies, cast it as a \emph{constrained Markov decision process} (CMDP), and then solve it based on the problem's dual form.

\subsection{Problem Formulation}
The objective is to maximize the expected \emph{discounted} sum of the updates' effectiveness in fulfilling the subscribed goal, where each agent individually derives its decision policy subject to the relevant ensued cost by looking into the problem from its own perspective. Let us define $\pi_\alpha^*$ and $\pi_\beta^*$ as the classes of optimal policies for transmission and query controls, respectively. Therefore, we can formulate the decision problem solved at each agent as follows\footnote{In an ideal scenario, where both agents have full knowledge of the goals and the source’s evolution, the problem could be approached in a centralized manner, allowing for the joint derivation of policies for both agents. However, this falls outside the scope of this work.}
\begin{align}\label{sec4:eq1}
    \mathcal{P}_1 : ~ &\underset{\pi_\gamma}{\operatorname{max}} ~\underset{N \rightarrow \infty}{\lim \sup} ~ \frac{1}{N}\, \mathbb{E}\bigg[ \sum_{n = 1}^{N} \lambda^n E_n\big| E_0\bigg] \nonumber \\
    & {\rm s.t.}~~\, \underset{N \rightarrow \infty}{\lim \sup} ~ \frac{1}{N}\, \mathbb{E}\bigg[ \sum_{n = 1}^{N} \lambda^n c_\gamma(\gamma_n)\bigg] \leq C_{\gamma, \rm max}
\end{align}
where $\lambda \in [0,1]$ indicates a discount factor, and $\gamma\in\{\alpha, \beta\}$ is replaced with $\alpha$ and $\beta$ for the update transmission and query decision problems, respectively, at the SA and the AA. Herein, $\gamma_n \in \{0,1\}$ denotes the decision at the relevant agent, $c_\gamma:\{0,1\} \rightarrow \mathbb{R}_0^+$ is a non-decreasing cost function, and $C_{\gamma, \rm max}$ shows the maximum discounted cost.

For either update communication model introduced in Section~\ref{sec2a}, optimal decisions at the agent(s) following the effect-aware policy, i.e., $\pi_\alpha^*$ and/or $\pi_\beta^*$, are obtained by solving $\mathcal{P}_1$ in \eqref{sec4:eq1}. However, for every agent that employs an effect-agnostic policy, with regard to Section~\ref{sec2b}, there is a predefined/given set of decisions denoted by $\tilde{\pi}_\alpha$ or $\tilde{\pi}_\beta$ such that $\pi_\alpha=\tilde{\pi}_\alpha$ or $\pi_\beta=\tilde{\pi}_\beta$, respectively.

\subsection{CMDP Modeling}\label{sec4b}
We cast $\mathcal{P}_1$ from \eqref{sec4:eq1} into an \emph{infinite-horizon} CMDP denoted by a tuple $(\mathcal{S}_\gamma, \mathcal{A}_\gamma, P_\gamma, r_\gamma)$ with components that are defined via the agent that solves the decision problem.

\subsubsection{Modeling at the SA} The CMDP at the SA is modeled according to the following components:

\textbf{\textit{States --}} The state of the system $S_{\alpha, n}$ at the $n$-th slot from the SA's perspective is depicted by a tuple $(v_n, \hat{E}_n)$ in which $v_n$ is the update’s usefulness, and $\hat{E}_n \in \{0, 1\}$ shows the E-ACK arrival status at the SA after passing the PEC, as defined in Section~\ref{sec2}. Herein, we have $\hat{E}_n=0$ in case $E_n=0$ or the acknowledgment signal is erased; otherwise, $\hat{E}_n=1$. In this regard, $S_{\alpha, n}$ belongs to a finite and countable state space $\mathcal{S}_\alpha$ with $\lvert\mathcal{S}_\alpha \rvert = 2\cdot\lvert\mathcal{V}\rvert$ elements.

\textbf{\textit{Actions --}} We denote $\alpha_n$ the decision for update communication at the $n$-th slot, which is a member of an action space $\mathcal{A}_\alpha = \{0,1\}$. In this space, $0$ stands for discarding the update, and $1$ indicates transmitting the update.

\textbf{\textit{Transition probabilities --}} The transition probability from the current state $S_{\alpha, n}$ to the future state $S_{\alpha, n+1}$ via taking the action $\alpha_n$ is written by 
\begin{align}\label{sec4:eq2}
    p_\alpha(S_{\alpha, n}, \alpha_n, S_{\alpha, n+1}) &= \mathrm{Pr}\big((v_{n+1}, \hat{E}_{n+1})\,\lvert\,(v_{n}, \hat{E}_n), \alpha_n\big) \nonumber \\
    &= p_\nu(v_{n+1})\mathrm{Pr}\big( \hat{E}_{n+1}\,\lvert\,v_{n}, \alpha_n\big)
\end{align}
since $\hat{E}_{n+1}$ and $\hat{E}_{n}$ are independent, and $\hat{E}_{n}$ is independent of $v_n$, $\forall n$. We can derive the conditional probability in \eqref{sec4:eq2} as
\begin{itemize}
    \item $\mathrm{Pr}\big( \hat{E}_{n+1}=0\,\lvert\,v_{n}, \alpha_n\big) = \mathrm{Pr}(\hat{v}_{\rm tgt} > \alpha_n v_{n} ) = 1 - P_{\hat{v}_{\rm tgt}}(\alpha_nv_n)$,
    \item $\mathrm{Pr}\big( \hat{E}_{n+1}=1\,\lvert\,v_{n}, \alpha_n\big) = P_{\hat{v}_{\rm tgt}}(\alpha_nv_n)$,
\end{itemize}
where $\hat{v}_{\rm tgt}$ is a \emph{mapped} target usefulness that the SA considers, $P_{\hat{v}_{\rm tgt}}(\hat{v}_{\rm tgt}) = \sum_{\hat{v}_{\rm tgt}^\prime \leq \hat{v}_{\rm tgt}} p_{\hat{v}_{\rm tgt}}(\hat{v}_{\rm tgt}^\prime)$ indicates its cumulative distribution function (CDF), and $p_{\hat{v}_{\rm tgt}}(\cdot)$ shows the pmf derived in Section~\ref{sec5b}. 

\textbf{\textit{Rewards --}} The immediate reward of moving from the state $S_{\alpha, n}$ to the state $S_{\alpha, n+1}$ under the action $\alpha_n$ is equal to $r_\alpha(S_{\alpha, n}, \alpha_n, S_{\alpha, n+1}) = \hat{E}_{n+1}$ where it relies on the E-ACK status in the future state. 

Despite possible erasures over the acknowledgment link, the reward defined in this model fits into the decision problem in \eqref{sec4:eq1}, where the corresponding objective becomes maximizing the expected discounted sum of E-ACK arrivals. In this sense, $\hat{E}_n$ at the SA resembles $E_n$ at the AA plus noise in the form of the E-ACK erasure.

\subsubsection{Modeling at the AA} For modeling the problem at the AA, we have the components as follows:

\textbf{\textit{States --}} We represent the state $S_{\beta, n}$ at the $n$-th time slot using a tuple $(\hat{v}_n, \Delta_n, \Theta_n)$, where $\hat{v}_n$ is the usefulness of the received update from the perspective of the endpoint, $\Delta_n$ is the AoI, and $\Theta_n$ denotes the action lateness, as modeled in Section~\ref{sec3}. Without loss of generality, we assume the values of $\Delta_n$ and $\Theta_n$ are truncated by the maximum values notated as $\Delta_{\rm max}$ and $\Theta_{\rm max}$, respectively, such that the conditions
\begin{align}
        g_\Delta(\hat{v}_n, \Delta_{\rm max}\!-\!1)\leq (1 +\varepsilon_\Delta) g_\Delta(\hat{v}_n, \Delta_{\rm max}),
\end{align}
for $\hat{v}_n \in \hat{\mathcal{V}}$, and
\begin{align}
        g_\Theta(\Theta_{\rm max}\!-\!1)\leq (1 +\varepsilon_\Theta) g_\Theta(\Theta_{\rm max})
\end{align}   
are met with the relevant accuracy $\varepsilon_\Delta$ and $\varepsilon_\Theta$. Given this, at the AA, $S_{\beta, n}$ is a member of a finite and countable space $\mathcal{S}_\beta$ having $\lvert\mathcal{S}_\beta \rvert = \Delta_{\rm max}\cdot\Theta_{\rm max}\cdot\lvert\hat{\mathcal{V}}\rvert$ states.

\textbf{\textit{Actions --}} As already mentioned, $\beta_n$ shows the decision of raising a query at the $n$-th time slot and gets values from an action space $\mathcal{A}_\beta = \{0,1\}$. Here, $0$ and $1$ depict refusing and confirming to pull an update, respectively.

\textbf{\textit{Transition probabilities --}} The transition probability from the current state $S_{\beta, n}$ to the future state $S_{\beta, n+1}$ under the action $\beta_n$ is modeled as  
\begin{align}\label{sec4:eq5}
        & p_\beta(S_{\beta, n}, \beta_n, S_{\beta, n+1})= \nonumber \\
        &~~~~~~~\mathrm{Pr}\big((\hat{v}_{n+1}, \Delta_{n+1}, \Theta_{n+1})\,\lvert\,(\hat{v}_{n}, \Delta_n, \Theta_n), \beta_n\big).
    \end{align}
According to \eqref{sec4:eq5}, we can write: 
\begin{itemize}
    \item $\mathrm{Pr}((\hat{\nu}_{j}, \operatorname{min}\{\Delta_{n}\!+\!1, \Delta_{\rm max}\},\\ \operatorname{min}\{\Theta_{n}\!+\!1, \Theta_{\rm max}\})\,\lvert\,(\hat{\nu}_{j}, \Delta_n, , \Theta_n), \beta_n) = 1-\beta_n$,
        
    \item $\mathrm{Pr}((\hat{\nu}_{j}, \operatorname{min}\{\Delta_{n}\!+\!1, \Delta_{\rm max}\}, 1)\,\lvert\,(\hat{\nu}_{j}, \Delta_n, \Theta_n), \beta_n) = \beta_n p_{\epsilon}$,
        
    \item $\mathrm{Pr}((\hat{\nu}_{j^\prime},1, 1)\,\lvert\,(\hat{\nu}_{j}, \Delta_n, \Theta_n), \beta_n) = \beta_n(1\!-\!p_{\epsilon})q_{j^\prime}$,
\end{itemize}
with $\hat{\nu}_{j}, \hat{\nu}_{j^\prime} \in \hat{\mathcal{V}}$. For the rest of the transitions, we have $p_\beta(S_{\beta, n}, \beta_n, S_{\beta, n+1}) = 0$. As stated earlier, $q_j = p_{\hat{\nu}}(\hat{\nu}_j)$ with the pmf $p_{\hat{\nu}}(\cdot)$ derived in Section~\ref{sec5a}.

\textbf{\textit{Rewards --}} Arriving at the state $S_{\beta, n+1}$ from the state $S_{\beta, n}$ by taking the action $\beta_n$, is rewarded based on the effectiveness level provided at the future state such that
\begin{align}
&r_\beta(S_{\beta, n}, \beta_n, S_{\beta, n+1}) = E_{n+1} = \nonumber \\ 
&~~~~~~\mathbbm{1}\big\{f_g(g_\Delta(\hat{v}_{n+1}, \Delta_{n+1}), g_\Theta(\Theta_{n+1})) \geq \operatorname{GoE}_{\rm tgt}\big\}  \nonumber \\
&~~~~~~\times \mathbbm{1}\big\{ \Theta_{n+1} < \Theta_{\rm max}\big\}
\end{align}
by the use of \eqref{sec3:eq1} and \eqref{sec3:eq4}.

\subsubsection{Independence of the initial state}
Before we delve into the dual problem and solve it, we state and prove two propositions to show that the expected discounted sum of effectiveness in \eqref{sec4:eq1} is the same for all initial states.

\begin{proposition}\label{propose1} The CMDP modeled at the SA satisfies the accessibility condition.
\end{proposition}
\begin{proof}
Given the transition probabilities defined in \eqref{sec4:eq2}, every state $S_{\alpha, m} \in \mathcal{S}_\alpha$, $m\leq N$, is accessible or reachable from the state $S_{\alpha, n} $ in finite steps with a non-zero probability, following the policy $\pi_\alpha$. Therefore, the accessibility condition holds for the CMDP modeled at the SA \cite[Definition~4.2.1]{bertsekas2007volume}.
\end{proof}

\begin{proposition}\label{propose2} The modeled CMDP at the AA meets the weak accessibility condition.
\end{proposition}
\begin{proof}
We divide the state space $\mathcal{S}_\beta$ into two disjoint spaces of $\mathcal{T}_a$ and $\mathcal{T}_b = \mathcal{S}_\beta - \mathcal{T}_a$, where $\mathcal{T}_a$ consists of all the states whose $\Delta_n=1$, i.e.,  $\mathcal{T}_a = \{S_{\beta,n} \,\lvert\, S_{\beta, n}=({\hat\nu}_j, 1, \Theta_n), \forall \hat{\nu}_j\in\hat{\mathcal{V}}, \Theta_n=1, 2, \dots, \Theta_{\rm max}\}$. Thus, $\mathcal{T}_b$ includes the rest of the states with $\Delta_n\geq 2$. With regard to the transition probabilities derived in \eqref{sec4:eq5}, all states of $\mathcal{T}_b$ are transient under any policy, while every state of an arbitrary pair of two states in $\mathcal{T}_a$ is accessible from the other state. Accordingly, the weak accessibility condition in the modeled CMDP at the AA is satisfied according to \cite[Definition~4.2.2]{bertsekas2007volume}.
\end{proof}

Given Propositions~\ref{propose1} and \ref{propose2}, we can show that the expected effectiveness obtained by $\mathcal{P}_1$ in \eqref{sec4:eq1} is the same for all initial states \cite[Proposition~4.2.3]{bertsekas2007volume}. In this regard, $E_n$, $\forall n$, is independent of $E_0$ for either model, thus we arrive at the following decision problem:
\begin{align}\label{sec4:eq7}
    \mathcal{P}_2 : ~ &\underset{\pi_\gamma}{\operatorname{max}} ~\underset{N \rightarrow \infty}{\lim \sup} ~ \frac{1}{N}\, \mathbb{E}\bigg[ \sum_{n = 1}^{N} \lambda^n E_n\bigg] \nonumber \\
    & {\rm s.t.}~~\, \underset{N \rightarrow \infty}{\lim \sup} ~ \frac{1}{N}\, \mathbb{E}\bigg[ \sum_{n = 1}^{N} \lambda^n c_\gamma(\gamma_n)\bigg] \leq C_{\gamma, \rm max}
\end{align}
for $\gamma=\{\alpha, \beta\}$. Applying Propositions~\ref{propose1} and \ref{propose2} confirms that there exist stationary optimal policies $\pi_\alpha^*$ and $\pi_\beta^*$ for $\mathcal{P}_2$ solved at the SA and the AA, respectively, where both policies are \emph{unichain} \cite[Proposition~4.2.6]{bertsekas2007volume}. 

\subsection{Dual Problem}
To solve the decision problem $\mathcal{P}_2$ given in \eqref{sec4:eq7}, we first define an unconstrained form for the problem via dualizing the constraint. Then, we propose an algorithm to compute the decision policies at both agents.

The unconstrained form of the problem is derived by writing the Lagrange function $\mathcal{L}(\mu; \pi_\gamma)$ as below
\begin{align}\label{sec4:eq8}
    \mathcal{L}(\mu; \pi_\gamma) &= ~\underset{N \rightarrow \infty}{\lim \sup} ~ \frac{1}{N}\, \mathbb{E}\bigg[\sum_{n = 1}^{N} \lambda^n\Big(E_n - \mu c_\gamma(\gamma_n)\Big)\bigg] \nonumber \\
    &~~~+ \mu C_{\gamma, \rm max} 
\end{align}
with $\mu \geq 0$ being the Lagrange multiplier. According to \eqref{sec4:eq8}, we arrive at the following dual problem to be solved:
\begin{equation}\label{sec4:eq9}
    \mathcal{P}_3 : ~ \underset{\mu \geq 0}{\inf}~ \underbrace{\underset{\pi_\gamma}{\operatorname{max}}~\mathcal{L}(\mu; \pi_\gamma)}_{\coloneqq\, h_\gamma(\mu)}
\end{equation}
where $h_\gamma(\mu)=\mathcal{L}(\mu; \pi_{\gamma, \mu}^*)$ is the Lagrange dual function with $\pi_{\gamma, \mu}^* \!: \mathcal{S}_\gamma \rightarrow \mathcal{A}_\gamma$ denoting a stationary $\mu$-optimal policy, which is obtained as
\begin{equation}\label{sec4:eq10}
    \pi_{\gamma, \mu}^* =  \underset{\pi_\gamma}{\operatorname{arg\,max}}~\mathcal{L}(\mu; \pi_\gamma)
\end{equation}
for $\mu$ derived in the dual problem $\mathcal{P}_3$. As the dimension of the state space $\mathcal{S}_\gamma$ is finite for both defined models, the growth condition is met \cite{altman1999constrained}. Furthermore, the immediate reward and the induced cost are bounded according to Section~\ref{sec4}. Given these satisfied conditions, from \cite[Corollary~12.2]{altman1999constrained}, we can assert that $\mathcal{P}_2$ and $\mathcal{P}_3$ converge to the same expected values such that we can write
\begin{equation}\label{sec4:eq11}
   \underset{\mu \geq 0}{\inf}~ \underset{\pi_\gamma}{\operatorname{max}}~\mathcal{L}(\mu; \pi_\gamma) =  \underset{\mu \geq 0}{\inf}~ \mathcal{L}(\mu; \pi_\gamma^*) = \underset{\pi_\gamma}{\operatorname{max}}~ \mathcal{L}(\mu^*; \pi_\gamma)
\end{equation}
for some policy class $\pi_\gamma^*$. Given the conditions satisfied, there exist non-negative optimal values for the Lagrange multiplier $\mu^*$ under Slater's condition such that \eqref{sec4:eq11} holds \cite[Theorem~12.8]{altman1999constrained}, establishing the strong duality between $\mathcal{P}_2$ and $\mathcal{P}_3$.

We can now proceed to derive the optimal policies at the SA and the AA from the decision problem $\mathcal{P}_3$ by applying an \emph{iterative algorithm} in line with the dynamic programming approach based on \eqref{sec4:eq8}--\eqref{sec4:eq10} \cite{hatami2022demand}. 

\subsection{Iterative Algorithm}
The iterative algorithm is given in Algorithm~\ref{Alg1} and consists of two \emph{inner} and \emph{outer} loops. The inner loop is for computing the $\mu$-optimal policy, i.e., $\pi_{\gamma, \mu}^*$, using the \emph{value iteration} method. Over the outer loop, the optimal Lagrange multiplier $\mu^*$ is derived via the \emph{bisection search} method.
\SetKwFunction{FMain}{$\text{policy}$} 
\SetKwProg{Fn}{Function}{:}{}
\begin{algorithm}[t!]
{
\DontPrintSemicolon
    \caption{Solution for deriving $\pi_{\gamma}^*$ and $\mu^*$} \label{Alg1}
    \KwInput{Given parameters $N \gg 1$, $C_{\gamma, \rm max}$, $\eta$, $\varepsilon_\mu$, CMDP's state space, i.e., $\mathcal{S}_\gamma$, and action space, i.e., $\mathcal{A}_\gamma$. The form of the cost function $c_\gamma(\cdot)$. Initial values $l \leftarrow 0$, $\mu^{(0)} \leftarrow 0$, $\mu_{-}^{(0)} \leftarrow 0$, $\mu_{+}^{(0)} > 0$, $\pi_{\gamma, \mu_{-}^{(0)}} \leftarrow 0$, and $\pi_{\gamma, \mu_{+}^{(0)}} \leftarrow 0$.
    }
    Derive $\pi_{\gamma, \mu}^*(s)$, $\forall s \in \mathcal{S}_\gamma$, via running \FMain{$\mu^{(0)}$}.\\
    \lIf{$\mathbb{E}\Big[\sum_{n = 1}^{N} c_\gamma(\gamma_n)\Big] \leq NC_{\gamma, \rm max}$}{\textbf{goto} {\scriptsize{\textbf{\ref{line:return1}}}}.}
    \While{$|\mu_{+}^{(l)}-\mu_{-}^{(l)}| \geq \varepsilon_\mu$}
    {
    \nonl\textit{{Step}} $l$: \Comment{\small \textit{Outer loop (Bisection search)}}\\
    set $l \leftarrow l+1$, $\mu_{-}^{(l)} \leftarrow \mu_{-}^{(l-1)}$, and $\mu_{+}^{(l)} \leftarrow \mu_{+}^{(l-1)}$.\\
    Update $\mu^{(l)} \leftarrow \frac{\mu_{-}^{(l)}+\mu_{+}^{(l)}}{2}$. \\
    Improve $\pi_{\gamma, \mu}^* \leftarrow$ \FMain{$\mu^{(l)}$}.\\
    \If{$\mathbb{E}\Big[\sum_{n = 1}^{N} c_\gamma(\gamma_n)\Big] \geq NC_{\gamma, \rm max}$}{$\mu_{-}^{(l)} \leftarrow \mu^{(l)}$, and $\pi_{\gamma, \mu_{-}^{(l)}} \leftarrow$ \FMain{$\mu_{-}^{(l)}$}.}
    \lElse{$\mu_{+}^{(l)} \leftarrow \mu^{(l)}$, and $\pi_{\gamma, \mu_{+}^{(l)}} \leftarrow$ \FMain{$\mu_{+}^{(l)}$}.}
    } 
    \lIf{$\mathbb{E}\Big[\sum_{n = 1}^{N} c_\gamma(\gamma_n)\Big] < NC_{\gamma, \rm max}$}{$\pi_{\gamma, \mu}^*(s) \leftarrow \eta \pi_{\gamma, \mu_{-}^{(l)}}(s) + (1-\eta)\pi_{\gamma, \mu_{+}^{(l)}}(s)$, $\forall s \in \mathcal{S}_\gamma$.}
    \KwRet $\mu^* = \mu^{(l)}$ and $\pi_\gamma^*(s) = \pi^*_{\gamma, \mu}(s)$, $\forall s\in\mathcal{S}_\gamma$.\label{line:return1}\\
    \nonl
    \nonl\hrulefill\\
    \nonl
    \Fn{\FMain{$\mu$}}{
    \KwInput{Known parameters from the outer loop. Initial values $k \leftarrow 1$, $\pi_{\gamma, \mu}(s) \leftarrow 0$, and $V_{k}^{\pi_{\gamma, \mu}}(s) \leftarrow 0$, $\forall s \in \mathcal{S}_\gamma$.
    }
        \nonl\textit{{Iteration}} $k$: \Comment{\small \textit{Inner loop (Value iteration)}}\\
        \For{state $s \in \mathcal{S}_\gamma$\label{line:iter_t}}{compute $V_{k}^{\pi_{\gamma, \mu}}(s)$ from \eqref{sec4:eq13}\label{line:compute_v}.\\
        Improve $\pi_{\gamma, \mu}(s)$ according to \eqref{sec4:eq14} and \eqref{sec4:eq15}.
        }
        \If{$\operatorname{sp}\!\left(V_{k}^{\pi_{\gamma, \mu}} - V_{k-1}^{\pi_{\gamma, \mu}}\right) \geq \varepsilon_\pi$ as in \eqref{sec4:eq17}}{step up $k \leftarrow k+1$, and \textbf{goto} {\scriptsize{\textbf{\ref{line:iter_t}}}}.}
        \KwRet $\pi_{\gamma, \mu}^*(s)=\pi_{\gamma, \mu}(s)$, $\forall s\in\mathcal{S}_\gamma$.
  }
}
\end{algorithm}

\subsubsection{Computing $\pi_{\gamma, \mu}^*$}
Applying the value iteration method, the decision policy is iteratively improved given $\mu$ from the outer loop (bisection search). Thus, $\pi_{\gamma, \mu}(s) \in \mathcal{A}_\gamma$, $\forall s \in \mathcal{S}_\gamma$, is updated such that it maximizes the \emph{expected utility} (value) $V_{k}^{\pi_{\gamma, \mu}}(s)$ at the $k$-th, $\forall k \in \mathbb{N}$, iteration, which is obtained as
\begin{align}\label{sec4:eq12}
    V_{k}^{\pi_{\gamma, \mu}}(s) &= \mathbb{E}\big[r_k+\lambda r_{k+1}+\lambda^2 r_{k+2}+\cdot\cdot\cdot\,\lvert\,s_k=s\big] \nonumber \\
    &\approx \mathbb{E}\big[r_k+\lambda V_{k-1}^{\pi_{\gamma, \mu}} \,\lvert\,s_k=s\big]
\end{align}
where $s_k$ denotes the state at the $k$-th iteration, and $r_k$ is the corresponding reward at that state. 
The approximation in \eqref{sec4:eq12} appears after bootstrapping the rest of the discounted sum of the rewards by the value estimate $V_{k-1}^{\pi_{\gamma, \mu}}$. Under the form of the value iteration for the unichain policy MDPs \cite{puterman2014markov}, the optimal value function is derived from Bellman's equation \cite{bellman1952theory}, as
\begin{align}\label{sec4:eq13}
    &V_{k}^{\pi_{\gamma, \mu}}(s) = \nonumber \\
    &~~~\underset{\gamma \in \mathcal{A}_\gamma}{\operatorname{max}}\sum_{s^\prime \in \mathcal{S}_\gamma}p_\gamma(s, \gamma, s^\prime) \Big[r_{\gamma, \mu}(s, \gamma, s^\prime) + \lambda V_{k-1}^{\pi_{\gamma, \mu}}(s^\prime) \Big]
\end{align}
for the state $s\in\mathcal{S}_\gamma$. Consequently, the decision policy in that state is improved by
\begin{align}\label{sec4:eq14}
    &\pi_{\gamma, \mu}(s) \in \nonumber \\
    &\underset{\gamma \in \mathcal{A}_\gamma}{\operatorname{arg\,max}}\sum_{s^\prime \in \mathcal{S}_\gamma}p_\gamma(s, \gamma, s^\prime) \Big[r_{\gamma, \mu}(s, \gamma, s^\prime) + \lambda V_{k-1}^{\pi_{\gamma, \mu}}(s^\prime) \Big].
\end{align}
In \eqref{sec4:eq13} and \eqref{sec4:eq14}, we define a \emph{net} reward function as 
\begin{align}\label{sec4:eq15}
    r_{\gamma, \mu}(s, \gamma, s^\prime)=r_{\gamma}(s, \gamma, s^\prime)-\mu c_\gamma(\gamma),
\end{align}
which takes into account the cost caused by the action taken.

The value iteration stops running at the $k$-th iteration once the following convergence criterion is met \cite{puterman2014markov}:
\begin{align}\label{sec4:eq16}
    \operatorname{sp}\!\left(V_{k}^{\pi_{\gamma, \mu}} - V_{k-1}^{\pi_{\gamma, \mu}}\right) < \varepsilon_\pi
\end{align}
where $\varepsilon_\pi >0$ is the desired convergence accuracy, and $\operatorname{sp}(\cdot)$ indicates a \emph{span} function $\mathbb{R}_0^+ \rightarrow \mathbb{R}_0^+$ given as
\begin{align}\label{sec4:eq17}
    \operatorname{sp}\!\left(V_{k^\prime}^{\pi_{\gamma, \mu}}\right) = \underset{s\in \mathcal{S}}{\operatorname{max}}\, V_{k^\prime}^{\pi_{\gamma, \mu}}(s) - \underset{s\in \mathcal{S}}{\operatorname{min}}\, V_{k^\prime}^{\pi_{\gamma, \mu}}(s)
\end{align}
by using the span seminorm for the arbitrary $k^\prime$-th iteration \cite[Section~6.6.1]{puterman2014markov}. 
As the decision policies are unichain and have aperiodic transition matrices, the criterion in \eqref{sec4:eq16} is satisfied after finite iterations for any value of $\lambda\in[0,1]$ \cite[Theorem~8.5.4]{puterman2014markov}.

\subsubsection{Computing $\mu^*$}
We leverage the bisection search method to compute the optimal Lagrange multiplier over multiple steps in the outer loop based on the derived $\pi_{\gamma, \mu}^*$ from the inner loop. Starting with an initial interval $[\mu_{-}^{(0)}, \mu_{+}^{(0)}]$ such that $h_\gamma(\mu_{-}^{(0)}) h_\gamma(\mu_{+}^{(0)}) <0$, the value of the multiplier at the $l$-th, $\forall l\in \mathbb{N}$, step is improved by $\mu^{(l)} = \frac{\mu_{-}^{(l)}+\mu_{+}^{(l)}}{2}$. As shown in Algorithm~\ref{Alg1}, at each step, the value of either $\mu_{-}^{(l)}$ or $\mu_{+}^{(l)}$ and the corresponding decision policy $\pi_{\gamma, \mu_{-}^{(l)}}$ or $\pi_{\gamma, \mu_{+}^{(l)}}$, respectively, are updated according to the cost constraint in \eqref{sec4:eq7} until a stopping criterion $\rvert \mu_{+}^{(l)}-\mu_{-}^{(l)} \rvert < \varepsilon_\mu$ is reached with the accuracy $\varepsilon_\mu$. Considering \eqref{sec4:eq8} and \eqref{sec4:eq9}, $h_\gamma(\mu)$ is a non-increasing function of $\mu$ in the unsettled region where the cost constraint has not yet been met. In this regard, the bisection method searches for the smallest Lagrange multiplier that guarantees the constraint. Also, one can show that $h_\gamma(\mu)$ denotes a Lipschitz continuous function with the Lipschitz constant as below
\begin{align*}
    \bigg\lvert\, C_{\gamma, \rm max} - \underset{N \rightarrow \infty}{\lim \sup} ~ \frac{1}{N}\, \mathbb{E}\bigg[ \sum_{n = 1}^{N} \lambda^n c_\gamma(\gamma_n)\bigg] \bigg\rvert.
\end{align*}
Therefore, the bisection search converges to the optimal value of $\mu$ within finite steps \cite[pp.~294]{Wood2009}.

After the outer loop stops running, we obtain a stationary \emph{deterministic} decision policy as $\pi^*_{\gamma} = \pi_{\gamma, \mu}$ if the following condition holds:
\begin{align}\label{sec4:eq19}
   \underset{N \rightarrow \infty}{\lim \sup} ~ \frac{1}{N}\, \mathbb{E}\bigg[ \sum_{n = 1}^{N} \lambda^n c_\gamma(\gamma_n)\bigg] = C_{\gamma, \rm max}.
\end{align}
Otherwise, the derived policy becomes \emph{randomized} stationary in the shape of mixing two deterministic policies $\pi_{\gamma, \mu_{-}^{(l)}}$ and $\pi_{\gamma, \mu_{+}^{(l)}}$ with probability $\eta\in[0, 1]$ \cite{beutler1985optimal}, 
where 
\begin{align*}
    \pi_{\gamma, \mu_{-}^{(l)}} = \underset{\mu \rightarrow \mu_{-}^{(l)}}{\lim} \pi_{\gamma, \mu}, ~~ \pi_{\gamma, \mu_{+}^{(l)}} = \underset{\mu \rightarrow \mu_{+}^{(l)}}{\lim} \pi_{\gamma, \mu}
\end{align*}
in the $l$-th step, after which the bisection search terminates. Hence, we can write
\begin{align}\label{sec4:eq20}
    \pi_{\gamma}^* \leftarrow \eta \pi_{\gamma, \mu_{-}^{(l)}} + (1-\eta)\pi_{\gamma, \mu_{+}^{(l)}},
\end{align}
which implies that the decision policy is randomly chosen as $\pi^*_{\gamma} = \pi_{\gamma, \mu_{-}^{(l)}} $ and $\pi^*_{\gamma} = \pi_{\gamma, \mu_{+}^{(l)}} $ with probabilities $\eta$ and $1-\eta$, respectively. In \eqref{sec4:eq20}, $\eta$ is computed such that the condition in \eqref{sec4:eq19} is maintained.

\subsubsection{Complexity analysis}
The value iteration approach in the inner loop is \emph{polynomial} with $\mathcal{O}(\rvert\mathcal{A}_\gamma\rvert\rvert\mathcal{S}_\gamma\rvert^2)$ arithmetic operations at each iteration. Besides, the bisection search takes $\lceil \log_2(\frac{\mu_{+}^{(0)}}{\epsilon_\mu})\rceil$ steps to reach the optimal Lagrange multiplier within the tolerance of $\varepsilon_\mu$, given the derived policy. Here, $\mu_{+}^{(0)}$ is the upper bound of the initial interval for the multiplier, with the lower bound fixed at zero.
Thereby, the overall complexity of Algorithm~\ref{Alg1}, in terms of the number of arithmetic operations across both loops, can be calculated as follows
\begin{align*}
\mathcal{O}\!\left(\frac{2\rvert\mathcal{S}_\gamma\rvert^2}{1-\lambda}\log\!\left(\frac{1}{1-\lambda}\right)\!\log\!\left(\frac{\mu_{+}^{(0)}}{\varepsilon_\mu}\right)\!\right)
\end{align*}
based on the fixed $\lambda < 1$ (see 
\cite{tseng1990solving} and \cite{littman2013complexity}). The algorithm's complexity increases as the state space expands, the initial interval for the multiplier widens, $\lambda \rightarrow 1$, and $\varepsilon_\mu \rightarrow 0$.

\section{Monte Carlo Probability Distribution Estimation}\label{sec5}
In this section, we leverage the \emph{Monte Carlo} estimation method to statistically compute the estimated pmfs of the received updates’ usefulness from the endpoint's perspective at the AA, and the mapped target usefulness at the SA. To this end, we consider a time interval in the format of an \emph{estimation horizon} (E-horizon), followed by a \emph{decision horizon} (D-horizon), as illustrated in Fig.~\ref{sec5:fig1}. The E-horizon is exclusively reserved for the estimation processes and has a length of $M$ time slots, which is sufficiently large to enable an accurate estimation. The D-horizon represents the long-term time horizon with the sufficiently large length of $N\gg1$ slots, as defined in Section~\ref{sec4}, during which the agents find and apply their (model) CMDP-based decision policies. 

Within the E-horizon, the SA does not make any decisions. Instead, it focuses on communicating updates at the highest possible rate while adhering to cost constraints. Once receiving these updates, the AA measures their usefulness, stores them in memory, and sends E-ACK signals for effective updates. The SA logs whether the E-ACK has been successfully received or not at every slot of the E-horizon. Finally, employing the received E-ACK signals at the SA and the measured updates’ usefulness at the AA, both agents perform their estimations.
\begin{figure}[t!]
    \centering
    \includegraphics[width=1\linewidth]{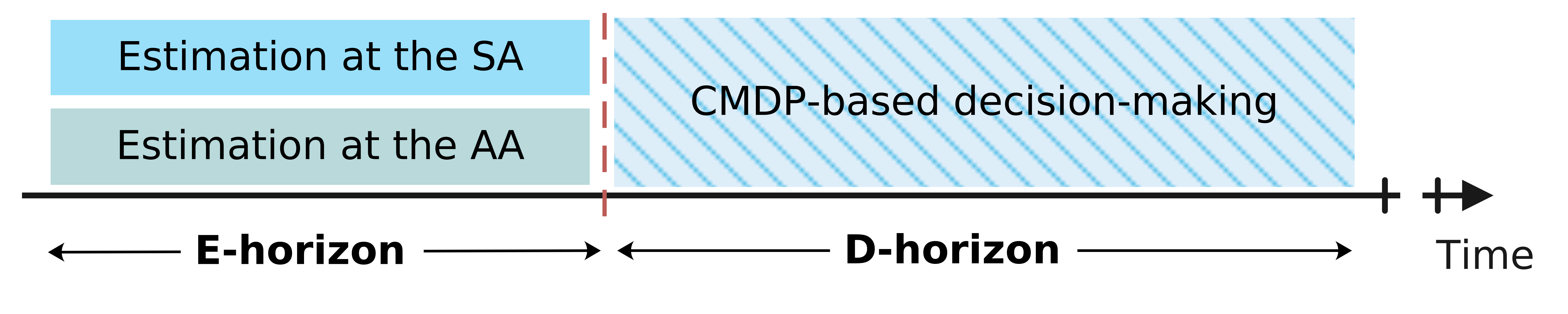}
    \vspace{-0.5cm}
    \caption{Time partitioning of the estimation and decision horizons.}
    \label{sec5:fig1}
\end{figure}

\subsection{Usefulness Probability of Received Updates}\label{sec5a}
Picking the $j$-th, $\forall j\in\mathcal{J}$, outcome from the set $\hat{\mathcal{V}}$ (defined in Section~\ref{sec3a}) that corresponds to the received update's usefulness from the endpoint's perspective at the $m$-th slot of the E-horizon, i.e., $\hat{v}_m$, the relevant estimated probability of that outcome is given by
\begin{align}
    q_j = p_{\hat{\nu}}(\hat{\nu}_j) = \frac{1}{M}\sum_{m=1}^{M} \mathbbm{1}\big\{\hat{v}_m=\hat{\nu}_j\big\}.
\end{align}

\subsection{Probability of the Mapped Target Usefulness}\label{sec5b}
We assume that the mapped target usefulness, i.e., $\hat{v}_{\rm tgt}$, is a member of the set $\hat{\mathcal{V}}_{\rm tgt}=  \{\vartheta_j\,\lvert\, j\in\mathcal{J}_{\rm tgt}\}$ with i.i.d. elements, where $\mathcal{J}_{\rm tgt}=\{1, 2, \dots, \lvert\hat{\mathcal{V}}_{\rm tgt}\rvert\}$, and the probability of the $j$-th element is equal to $p_{\hat{v}_{\rm tgt}}(\hat{v}_{\rm tgt}=\vartheta_j)$. Herein, as mentioned earlier, $p_{\hat{v}_{\rm tgt}}(\cdot)$ is the estimated pmf of the mapped target usefulness and obtained by
\begin{align}\label{sec5b:eq1}
    p_{\hat{v}_{\rm tgt}}(\vartheta_j) &= \sum_{e\in\{0,1\}}p_{\hat{v}_{\rm tgt}}\big(\vartheta_j\,\lvert\,\hat{E}=e\big)\mathrm{Pr}\big(\hat{E}=e\big) 
\end{align}
where we find the probability of successfully receiving E-ACK, i.e., $e=1$, or not, i.e., $e=0$, as follows
\begin{align}\label{sec5b:eq2}
    \mathrm{Pr}\big(\hat{E}=e\big) = \frac{1}{M}\sum_{m=1}^{M} \mathbbm{1}\big\{\hat{E}_m=e\big\}
\end{align}
where $\hat{E}_m$ indicates the E-ACK arrival status at the $m$-th slot of the E-horizon.
Furthermore, to derive the conditional probability in \eqref{sec5b:eq1}, we first consider the successful arrivals of E-ACK signals such that
\begin{align}\label{sec5b:eq3} 
    p_{\hat{v}_{\rm tgt}}\big(\vartheta_j\,\lvert\,\hat{E}=1\big) &= \frac{\sum_{i\in\mathcal{I}} p_{\nu|\hat{E}=1}\big(\nu_i\,\lvert\,\nu_i\geq\vartheta_j\big)}{\sum_{j\in\mathcal{J}_{\rm tgt}}\sum_{i\in\mathcal{I}} p_{\nu|\hat{E}=1}\big(\nu_i\,\lvert\,\nu_i\geq\vartheta_j\big)}.
\end{align}
Then, we have
\begin{align}\label{sec5b:eq4} 
    p_{\hat{v}_{\rm tgt}}\big(\vartheta_j\,\lvert\,\hat{E}=0\big) &= \frac{\sum_{i\in\mathcal{I}} p_{\nu|\hat{E}=0}\big(\nu_i\,\lvert\,\nu_i<\vartheta_j\big)}{\sum_{j\in\mathcal{J}_{\rm tgt}}\sum_{i\in\mathcal{I}} p_{\nu|\hat{E}=0}\big(\nu_i\,\lvert\,\nu_i<\vartheta_j\big)}.
\end{align}
The pmfs $p_{\nu|\hat{E}=1}(\cdot)$ and $p_{\nu|\hat{E}=0}(\cdot)$ in \eqref{sec5b:eq3} and \eqref{sec5b:eq4} are associated with an observation's importance rank given the successful and unsuccessful communication of the E-ACK, respectively. In this regard, by applying Bayes' theorem we can derive the following formula:
\begin{align}\label{sec5b:eq5}
    p_{\nu|\hat{E}=e}\big(\nu_i\big)=\frac{\frac{1}{M}\sum_{m=1}^{M} \mathbbm{1}\big\{v_m=\nu_i\land \hat{E}_m=e\big\}}{\mathrm{Pr}(\hat{E}=e)}.
\end{align}

\section{Simulation Results}
In this section, we present simulation results that corroborate our analysis and assess the performance gains in terms of effectiveness achieved by applying different update models and agent decision policies in end-to-end status update systems. 

\subsection{Setup and Assumptions}
We study the performance over $5\times10^5$ time slots, which includes the E-horizon and the D-horizon with $1\times10^5$ and $4\times10^5$ slots, respectively. To model the Markovian effect-agnostic policy, we consider a Markov chain with two states, $0$ and $1$. We assume that the self-transition probability of state $0$ is $0.9$, while the one for state $1$ relies on the controlled update transmission or query rate. Without loss of generality, we assume that the outcome spaces for the usefulness of generated updates, i.e., $\mathcal{V}$, the usefulness of received updates, i.e., $\hat{\mathcal{V}}$, and the mapped target usefulness, i.e., $\hat{\mathcal{V}}_{\rm tgt}$, are bounded within the span $[0, 1]$. For simplicity, we divide each space into discrete levels based on its number of elements in ascending order, where every level shows a randomized value. We also consider that the $i$-th outcome of the set $\mathcal{V}$ notated as $\nu_i$, $i\in\mathcal{I}$, occurs following a \emph{beta-binomial} distribution with pmf
\begin{align}
    p_\nu(\nu_i) = \binom{\lvert\mathcal{V}\rvert-1}{i-1} \frac{\operatorname{Beta}(i-1+a, \lvert\mathcal{V}\rvert-i+b)}{\operatorname{Beta}(a, b)}
\end{align}
where $\operatorname{Beta}(\cdot, \cdot)$ is the beta function, and $a=0.3$ and $b=0.3$ are shape parameters.

\noindent
\renewcommand{\arraystretch}{1}
\begin{table}[!t]
\begin{center}
\caption{Parameters for simulation results}\label{sim:tab1}
\begin{tabular}{| l | c | c |}
\hline
\!\footnotesize \textbf{Name} & \!\footnotesize \textbf{Symbol}\! &\footnotesize \textbf{Value}\\
\hline
\hline
\footnotesize \!E-horizon length &\footnotesize -- &\footnotesize \!\!$1\!\times\!10^5\,[\text{slot}]$\!\!\\
\hline
\footnotesize \!D-horizon length &\footnotesize $N$ &\footnotesize  \!\!$4\!\times\!10^5\,[\text{slot}]$\!\!\\
\hline
\footnotesize \!Erasure probability in update channel &\footnotesize $p_{\epsilon}$ &\footnotesize $0.2$\\
\hline
\footnotesize \!Erasure probability in acknowledgment link &\footnotesize $p_{\epsilon}^\prime$ &\footnotesize $0.1$\\
\hline
\footnotesize \!Length of generated update usefulness space &\footnotesize $\lvert\mathcal{V}\rvert$ &\footnotesize $10$\\
\hline
\footnotesize \!Length of received update usefulness space &\footnotesize $\lvert\hat{\mathcal{V}}\rvert$ &\footnotesize $11$\\
\hline
\footnotesize \!Length of mapped target usefulness space &\footnotesize $\lvert\hat{\mathcal{V}}_{\rm tgt}\rvert$ &\footnotesize $11$\\
\hline
\multirow{2}{*}{\footnotesize \!Shape parameters for usefulness distribution} &\footnotesize $a$ &\footnotesize $0.3$\\
\cline{2-3}
&\footnotesize $b$ &\footnotesize $0.3$\\
\hline
\footnotesize \!Maximum truncated AoI &\footnotesize \!$\Delta_{\rm max}$\! &\footnotesize $10\,[\text{slot}]$\\
\hline
\footnotesize \!Action window width in pull-based model &\multirow{3}{*}{ \footnotesize \!$\Theta_{\rm max}$\!} & \footnotesize $1\,[\text{slot}]$\\
\cline{1-1}\cline{3-3}
\footnotesize \!Action window width in push-and-pull model\! & \footnotesize & \footnotesize $5\,[\text{slot}]$\\
\cline{1-1}\cline{3-3}
\footnotesize \!Action window width in push-based model & \footnotesize & \footnotesize $10\,[\text{slot}]$\\
\hline
\footnotesize \!Update transmission cost at $n$-th slot &\footnotesize $C_{n,1}$ &\footnotesize $0.1$\\
\hline
\footnotesize \!Query raising cost at $n$-th slot &\footnotesize $C_{n,2}$ &\footnotesize $0.1$\\
\hline
\footnotesize \!Actuation availability cost at $n$-th slot &\footnotesize $C_{n,3}$ &\footnotesize $0.01$\\
\hline
\footnotesize \!Maximum discounted cost in decision problem\! &\footnotesize \!$C_{\gamma, \rm max}$\! &\footnotesize $0.08$\\
\hline
\footnotesize \!Target effectiveness grade &\footnotesize \!$\operatorname{GoE}_{\rm tgt}$\! &\footnotesize $0.6$\\
\hline
\footnotesize \!Discount factor in CMDP &\footnotesize $\lambda$ &\footnotesize $0.75$\\
\hline
\multirow{2}{*}{\footnotesize \!Convergence accuracy in Algorithm~\ref{Alg1}} &\footnotesize $\varepsilon_\mu$ &\footnotesize $10^{-4}$\\
\cline{2-3}
&\footnotesize $\varepsilon_\pi$ &\footnotesize $10^{-4}$\\
\hline
\footnotesize \!Mixing probability in bisection method &\footnotesize $\eta$ &\footnotesize $0.5$\\
\hline
\footnotesize \!Controlled transmission rate (effect-agnostic) &\footnotesize -- &\footnotesize $0.8$\\
\hline
\footnotesize \!Controlled query rate (effect-agnostic) &\footnotesize -- &\footnotesize $0.8$\\
\hline
\end{tabular}
\medskip
\end{center}
\end{table}
Moreover, we plot the figures based on the following form of the GoE metric\footnote{The analysis can be easily extended to any other forms of the GoE metric.}, which comes from the general formulation proposed in \eqref{sec3:eq1}: 
\begin{align}\label{sim:eq1}
    \operatorname{GoE}_n = \frac{\hat{v}_n}{\Delta_n\Theta_n} - \alpha_n C_{n, 1} - \beta_n C_{n, 2} - C_{n, 3}
\end{align}
for the $n$-th time slot. Herein, $C_{n, 1}$ is the communication cost, $C_{n, 2}$ denotes the query cost, and $C_{n, 3}$ indicates the actuation availability cost at the AA which depends on the update communication model. However, the cost function in the decision problem's constraint is assumed to be $c_\gamma(\gamma_n) = \gamma_n$, $\forall n\in\mathbb{N}$. The parameters used in the simulations are summarized for Table~\ref{sim:tab1}. In the legends of the plotted figures, we depict the policies at the agents in the form of a tuple, with the first and second elements referring to the policy applied in the SA and the AA, respectively. For an effect-aware policy, we simply use the notation ``E-aware,'' while for an effect-agnostic policy, only the modeling process is mentioned in the legend.

\subsection{Results and Discussion}
Fig.~\ref{results:fig1} illustrates the evolution of the average cumulative effectiveness over time for the push-and-pull model and different agent decision policies. As it is shown, applying the effect-aware policy in both agents offers the highest effectiveness, where the gap between its and the other policies’ performance increases gradually as time passes. Nevertheless, using the other effect-agnostic policies at either or both agent(s) diminishes the effectiveness performance of the system by at least $12\%$ or $36\%$, respectively. Particularly, if the SA and the AA apply the effect-aware and Markovian effect-agnostic policies, sequentially, the offered effectiveness is even lower than the scenario in which both agents use periodic effect-agnostic policies. One of the reasons is that the Markovian query raising misleads the SA in estimating the mapped target usefulness of updates. It is also worth mentioning that throughout the E-horizon from slot $1$ to $1\times 10^5$, the scenarios where the AA initiates effect-aware queries offer the same performance and better than those of the other scenarios. However, in the D-horizon, a performance gap appears and evolves depending on the applied policy at the SA. Therefore, by making effect-aware decisions at both agents based on estimations in the E-horizon and CMDP-based policies in the D-horizon, performance can be significantly improved. 
\begin{figure}[t!]
    \centering
    \includegraphics[width=0.4\textwidth]{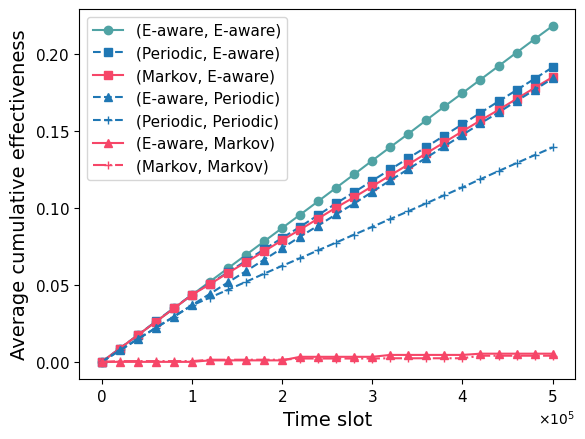}
    \caption{Evolution of the average effectiveness accumulated over time in the push-and-pull model.}
    \label{results:fig1}
\end{figure}

The bar chart in Fig.~\ref{results:figbar} shows the average rate of update transmissions from the SA and the consequent actions performed at the AA based on the initiated queries that result in the system’s effectiveness, as depicted in Fig.~\ref{results:fig1}. Interestingly, the scenario where both agents apply the effect-aware policy has lower transmission and action rates compared to the other scenarios except the ones with Markovian effect-agnostic query policies. However, the number of actions that can be taken for those scenarios with Markovian queries is limited since most updates are received outside the action windows. Figs.~\ref{results:fig1} and \ref{results:figbar} show that using effect-aware policies at both the SA and the AA not only brings the highest effectiveness but also needs lower update transmissions by an average of $11\%$, saving resources compared to the scenarios that have comparable performances. We also note that although effect-aware and periodic query decisions with effect-aware update transmission have almost the same action rate, the effect-aware case leads to more desirable effects or appropriate actions at the endpoint. This results in around $16\%$ higher average cumulative effectiveness, referring to Fig.~\ref{results:fig1}. 
\begin{figure}[t!]
    \centering
    \includegraphics[width=0.4\textwidth]{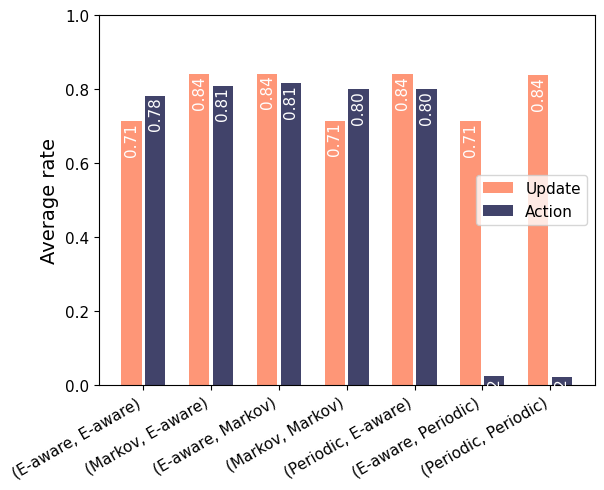}
    \caption{Average update transmission and query rates for different decision policies in the push-and-pull model.}
    \label{results:figbar}
\end{figure}

\begin{figure}[t!]
    \centering
    \subfloat[]{
    \includegraphics[width=0.4\textwidth]{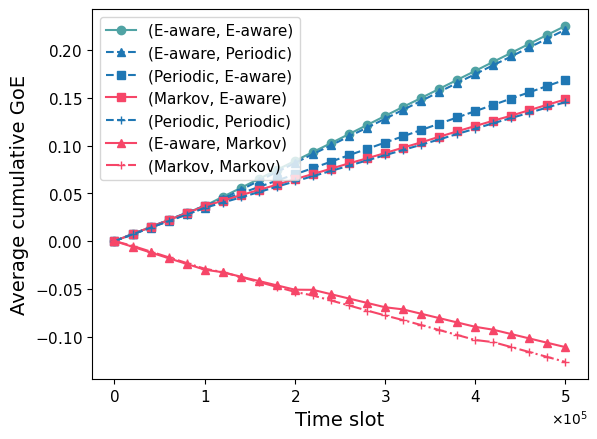} \label{results:fig2a}
    }
    \hfil
    \subfloat[]{
    \includegraphics[width=0.4\textwidth]{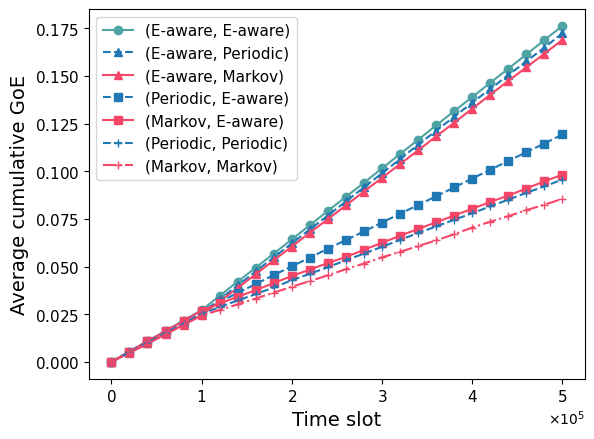} \label{results:fig2b}
    }
    \hfil
    \subfloat[]{
    \includegraphics[width=0.4\textwidth]{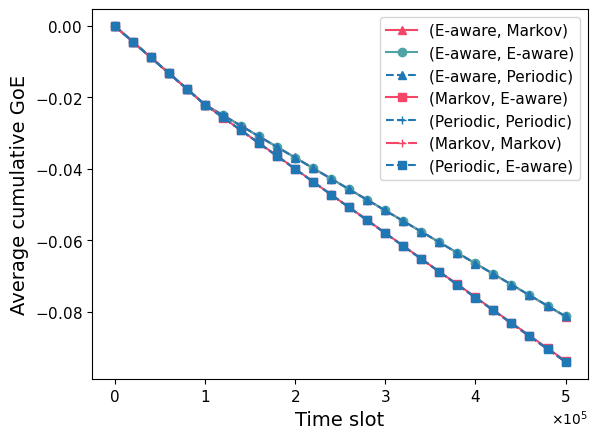} \label{results:fig2c}
    }
    \caption{Temporal evolution of the average cumulative GoE for the (a) push-and-pull, (b) push-based, and (c) pull-based models.}
    \label{results:fig2}
\end{figure}
Figs.~\ref{results:fig2a}, \ref{results:fig2b}, and \ref{results:fig2c} present the average cumulative GoE provided in the system over $5\times10^5$ time slots for the push-and-pull, push- and pull-based communication models, respectively. The corresponding effectiveness at the endpoint for the primary model can be found in Fig.~\ref{results:fig1}. The plots demonstrate that when both agents decide based on effect-aware policies, regardless of the update communication model, the highest offered GoE of the system is reached. However, in other scenarios, the performance of some policies exceeds those of others, depending on the update model. For instance, having effect-aware update transmission and Markovian queries shows a $2.52$ times higher average GoE for the push-based model than the push-and-pull one. This is because the push-based model has a larger action window.

Besides, comparing Figs.~\ref{results:fig2a} and \ref{results:fig2b}, the reason that the provided GoE by the effect-aware decisions at both agents is $28\%$ lower for the push-based model compared to the push-and-pull is that the AA has to be available longer, which causes a higher cost. Also, with the pull-based model as in Fig.~\ref{results:fig2c}, the average GoE within a period is less than the average cost since the AA is only available to act at query instants, significantly reducing the average GoE despite the high update transmission rate. In the pull-based model, however, applying the CMDP-based update transmission decisions at the SA can address this issue with a $16\%$ higher average GoE.

The trade-off between the average effectiveness and the width of the action window, i.e., $\Theta_{\rm max}$, is shown in Fig.~\ref{results:fig10} for different decision policies. We see that the system cannot offer notable effectiveness with $\Theta_{\rm max}=1$, i.e., under the pull-based model. However, by expanding the action window width from $\Theta_{\rm max}=1$ to $10\,[\text{slot}]$, the average effectiveness boosts from its lowest to the highest possible value. Since all policies already reach their best performance before $\Theta_{\rm max}=10$, which indicates the push-based model, we can conclude that the push-and-pull model with a flexible action window is more advantageous than the push-based one with a very large fixed window. In addition, the scenario where both agents use effect-aware policies outperforms the others for $\Theta_{\rm max}>1$. It is worth mentioning that the performance of the scenarios where the AA initiates Markovian queries converges to the highest level for very large widths. The average cumulative effectiveness of these scenarios, as in Fig.~\ref{results:fig1}, visibly rises for $\Theta_{\rm max}\geq9$. 
\begin{figure}[t!]
    \centering
    \includegraphics[width=0.4\textwidth]{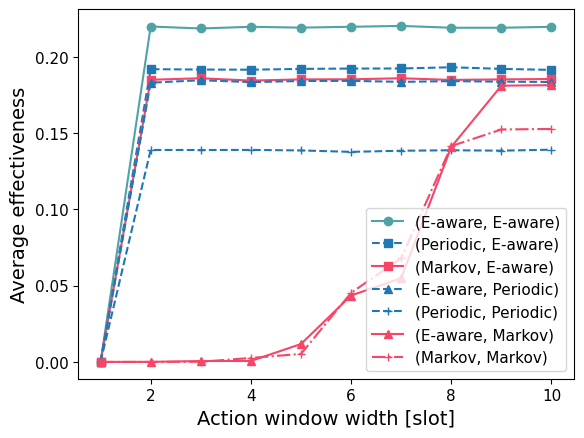}
    \caption{Average effectiveness versus the action window width in the push-and-pull communication model.}
    \label{results:fig10}
\end{figure}

\begin{figure}[t!]
    \centering
    \includegraphics[width=0.4\textwidth]{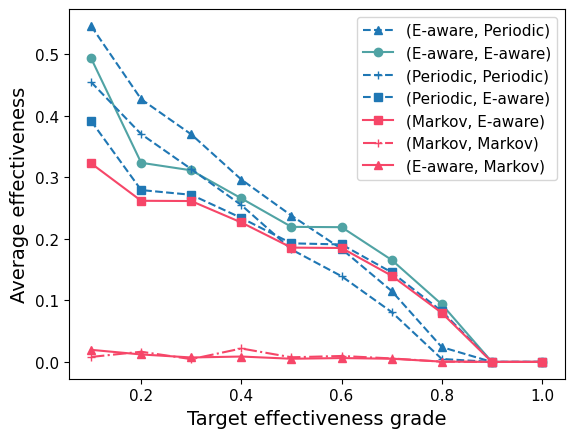}
    \caption{Average effectiveness as a function of the target effectiveness grade in the push-and-pull model.}
    \label{results:fig5}
\end{figure}
The interplay between the average effectiveness and the target effectiveness grade, i.e., $\operatorname{GoE}_{\rm tgt}$ from Section~\ref{sec3c}, is depicted in Fig.~\ref{results:fig5} for the push-and-pull model and different agent decision policies. Evidently, the average effectiveness under all policies decreases gradually with the increase of the target effectiveness grade and converges to zero for $\operatorname{GoE}_{\rm tgt}\geq 0.9$. Using effect-aware policies at both agents offers the highest effectiveness for medium-to-large target grades, i.e., $\operatorname{GoE}_{\rm tgt}\geq0.52$ here. However, for lower target grades, the effect-aware and periodic effect-agnostic decisions at the SA and the AA, respectively, result in better performance. This comes at the cost of higher transmission and action rates. Thus, there is a trade-off between the paid cost and the offered effectiveness, thus various policies can be applied depending on the cost criterion and the target effectiveness grade.

Fig.~\ref{results:fig6} depicts the average effectiveness obtained in the system through $5\times 10^5$ time slots versus the controlled update transmission rate for the push-and-pull model and various agent decision policies. Concerning Section~\ref{sec2b}, this controlled rate is related to the effect-agnostic update transmission policies at the SA and denotes the expected number of updates to be communicated within the specified period. Therefore, the performance of the other policies should remain fixed for different controlled transmission rates. Fig.~\ref{results:fig6} reveals that increasing the controlled update rate increases the offered effectiveness when the SA applies the effect-agnostic policies. However, even at the highest possible rate, subject to the maximum discounted cost, using effect-aware policies at both agents is necessary to ensure the highest average effectiveness, regardless of the controlled transmission rate. As an illustrative example, when the AA initiates effect-aware queries, the effectiveness drops by an average of $38\%$ ($43\%$) if the SA transmits updates based on the periodic (Markovian) effect-agnostic policy instead of using the effect-aware one. 
\begin{figure}[t!]
    \centering
    \includegraphics[width=0.4\textwidth]{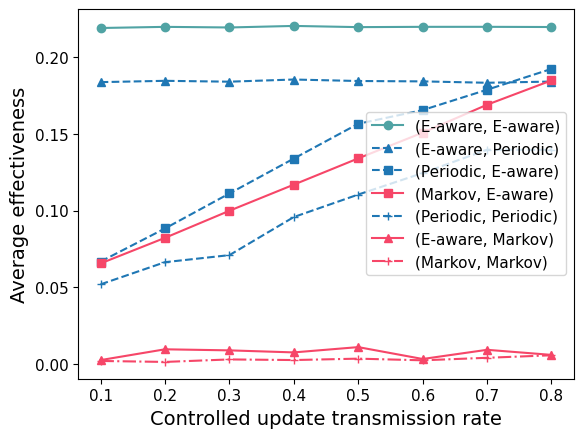}
    \caption{Comparison between different policies with variable update transmission rates but fixed query rates for effect-agnostic policies.}
    \label{results:fig6}
\end{figure}

In Fig.~\ref{results:fig7}, the plot shows the same trend as Fig.~\ref{results:fig6}, but this time it focuses on the changes in the average effectiveness versus the controlled query rate. Herein, the controlled query rate is dedicated to the scenarios where the AA operates under effect-agnostic policies. The results indicate that the average effectiveness rises via the increase of the controlled query rate for the periodic and Markovian policies. However, the increase in effectiveness is not significant for the latter one. Despite this, even with the highest rates, the effectiveness offered in the course of effect-aware queries is still higher than those of raising effect-agnostic queries. Therefore, the highest average effectiveness is achieved when both agents make effect-aware decisions, as depicted in Figs.~\ref{results:fig6} and \ref{results:fig7}. 
\begin{figure}[t!]
    \centering
    \includegraphics[width=0.4\textwidth]{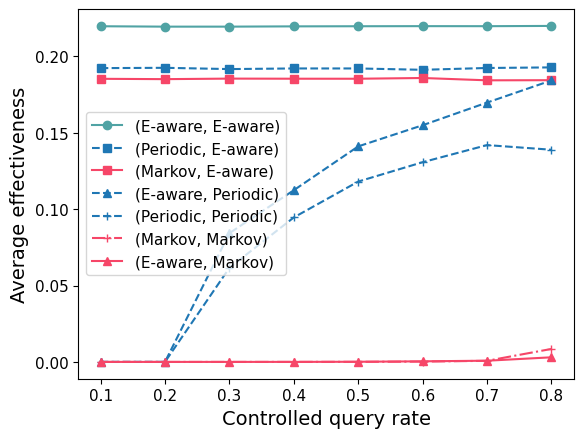}
    \caption{Comparison between different policies with variable query rates but fixed update transmission rates for effect-agnostic policies.}
    \label{results:fig7}
\end{figure}

In the context of the decision-making problem~$\mathcal{P}_2$ in \eqref{sec4:eq7}, altering the maximum discounted cost, i.e., $C_{\gamma, \rm max}$, can impact the decisions made by each agent. To study this, we have plotted Fig.~\ref{results:fig8} for the push-and-pull model under different decision policies. The figure shows that the stricter the cost constraint, the lower the average effectiveness, irrespective of the decision policy. Also, for all cost constraints, the scenario in which both agents use effect-aware policy yields the best performance, whereas the other policies outperform each other under different cost constraints. Due to CMDP-based decisions, the gap between the effectiveness performance of the best scenario and those of the others increases as the constraint decreases until $C_{\gamma, \rm max}=0.1$, where the SA can transmit all updates, and the AA can initiate queries without restrictions.
\begin{figure}[t!]
    \centering
    \includegraphics[width=0.4\textwidth]{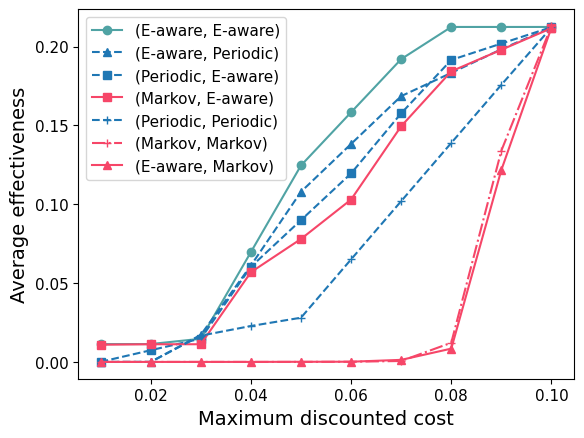}
    \caption{Average effectiveness achieved through different policies versus the maximum discounted cost in the push-and-pull model.}
    \label{results:fig8}
\end{figure}

According to the CMDP modeled at either agent, model-based effect-aware policies heavily rely on the conditions of the update and/or E-ACK channels. To evaluate this impact, Fig.~\ref{results:fading_channel} illustrates the average effectiveness as a function of the channel erasure probability. For clarity, we assume $p_{\epsilon}= p_{\epsilon}^\prime$, where $p_{\epsilon}$ and $p_{\epsilon}^\prime$ denote the erasure probabilities in the forward communication and the acknowledgment links, respectively. At lower erasure probabilities, particularly when $p_{\epsilon}\leq 0.52$, employing effect-aware policies at both agents yields the highest average effectiveness. However, as channel conditions worsen, the performance of effect-aware decision-making gradually declines. Notably, when $p_{\epsilon}\geq 0.66$, the effectiveness of pulling updates using effect-aware policies drops to its lowest levels.
\begin{figure}[t!]
    \centering
    \includegraphics[width=0.4\textwidth]{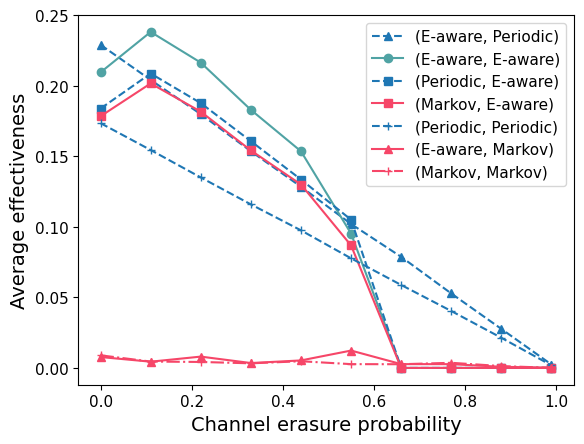}
    \caption{Interplay between average effectiveness and erasure probability for both the update and E-ACK channels in the push-and-pull model.}
    \label{results:fading_channel}
\end{figure}

Afterward, we compare the performance provided by model-based agent decisions discussed in Section~\ref{sec4} with that of \emph{model-free} decisions. The latter is based on reinforcement learning (RL), where each agent separately learns to make decisions through direct interaction with the environment. This learning process is modeled under the state (here, observation) spaces, action sets, and rewards according to Section~\ref{sec4b}, without relying on the construction of a predefined model. To derive model-free decisions, we employ a \emph{deep Q-network} (DQN) and parameterize an approximate value function for every agent within the E-horizon through a multilayer perceptron (MLP), assisted with the experience replay mechanism. The neural network consists of two hidden layers, each with $64$ neurons, and is trained using the adaptive moment estimation (Adam) optimizer. The default values for the RL setting are taken from \cite{mnih2015human}, except for the learning rate that is $10^{-4}$, and the discount factor assumed $0.75$, as aligned with Table~\ref{sim:tab1}. In this regard, Fig.~\ref{RLresults:fig1} depicts the evolution of the average cumulative effectiveness over time for the push-and-pull model. The graph compares the performance achieved by model-based agent decisions with that of model-free (RL-based) ones. As shown, the model-based approach outperforms the model-free approach across all decision policies. However, the performance gap is less significant when both agents apply effect-aware policies. Considering the CMDP modeled at both agents, several factors may contribute to the higher effectiveness of the model-based approach, such as the relatively small state space, especially for the modeling at the SA.
\begin{figure}[t!]
    \centering
    \includegraphics[width=0.4\textwidth]{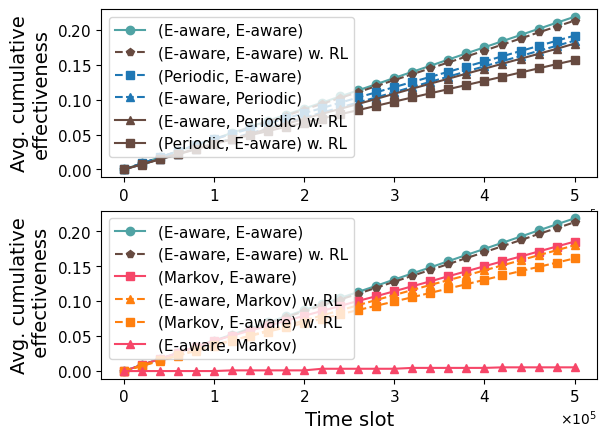}
    \caption{Temporal evolution of average cumulative effectiveness under model-based and model-free (RL-based) policies in the push-and-pull model.}
    \label{RLresults:fig1}
\end{figure}

Additionally, the length of the E-horizon potentially influences the performance and complexity of both decision-making approaches. To analyze this impact, Fig.~\ref{results:fig9} demonstrates how the E-horizon length affects the average effectiveness of effect-agnostic and effect-aware policies under both model-based and model-free approaches for the push-and-pull model. The results show that increasing the E-horizon length enhances performance, reaching its peak at an optimal E-horizon length. This optimal length value varies depending on the policy and approach employed. For instance, when both agents employ effect-aware policies, achieving the highest effectiveness requires approximately $10^{2.9}$ and $10^{3.7}$ slots for the model-based and model-free approaches, respectively, with convergence accuracy specified in Table~\ref{sim:tab1}. In this context, the required E-horizon length serves as a measure of the complexity of the corresponding algorithm for each approach. For the above example, the model-based algorithm (Algorithm~\ref{Alg1}) achieves a complexity that is $6.31$ times lower compared to the model-free DQN algorithm.
\begin{figure}[t!]
    \centering
    \includegraphics[width=0.4\textwidth]{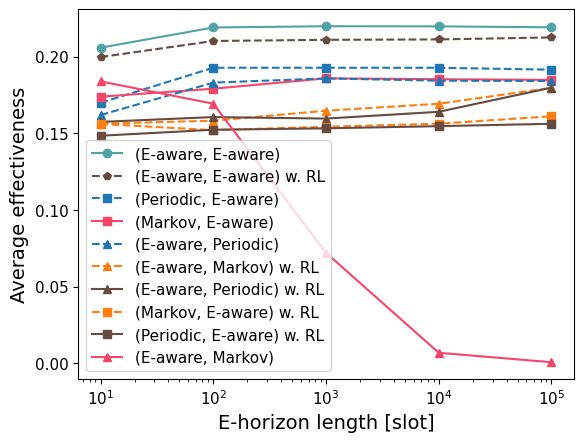}
    \caption{Impact of E-horizon length on the effectiveness of model-based and model-free (RL-based) policies in the push-and-pull model.}
    \label{results:fig9}
\end{figure}

\subsection{Lookup Maps for Agent Decisions}
\begin{figure}[t!]
    \centering
    \subfloat[]{
    \includegraphics[width=0.201\textwidth]{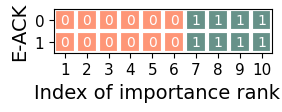}
    \label{th:fig10a}
    }
    \hfil
    \subfloat[]{
    \includegraphics[width=0.201\textwidth]{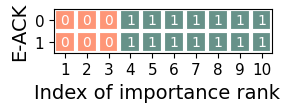}
    \label{th:fig10b}
    }  
    \caption{Lookup maps for decision-making at the SA with (a) $C_{\alpha, \rm max}=0.06$ and (b) $C_{\alpha, \rm max}=0.08$.}
    \label{th:fig10}
\end{figure}
\begin{figure}[t!]
    \centering
    \subfloat[]{
    \includegraphics[width=0.4\textwidth]{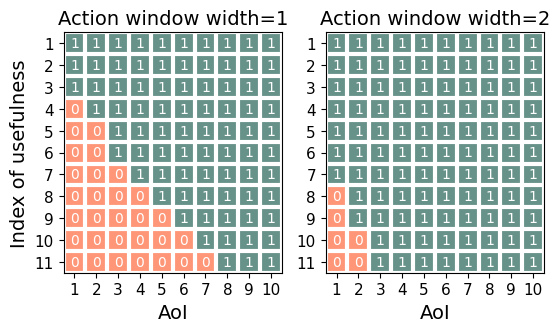}
    \label{th:fig11a}
    }
    \hfil
    \subfloat[]{
    \includegraphics[width=0.212\textwidth]{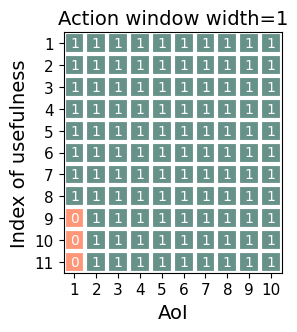}

    \label{th:fig11b}
    }
    \caption{Lookup maps for decision-making at the AA with (a) $C_{\beta, \rm max}=0.06$ and (b) $C_{\beta, \rm max}=0.08$. }
    \label{th:fig11}
\end{figure}
As the modeled CMDPs in Section~\ref{sec4b} have finite states, we can depict optimal mode-based decisions derived in Algorithm~\ref{Alg1} via a multi-dimensional \emph{lookup map} for each decision-making agent. Figs.~\ref{th:fig10} and \ref{th:fig11} illustrate the maps for the SA and the AA, respectively, under different maximum discounted costs. The number of dimensions in a map relies on the number of elements constructing every state of the relevant CMDP, with each dimension assigned to one element. With the lookup map in hand, an agent can make optimal decisions at each slot based on its current state. When comparing the same maps for different maximum discounted costs, the observation emerges that the more stringent the cost constraint is, the narrower the agent decision boundaries become. Thus, the maps could vary by changes in the parameters’ values or the goals with different target effectiveness grades. It is noteworthy that in Fig.~\ref{th:fig11}, the maps with $\Theta_n=1$ represent the pull-based model, while the push-and-pull model converges to the push-based model with $\Theta_n\geq3$ and $\Theta_n\geq2$ in Figs.~\ref{th:fig11a} and \ref{th:fig11b}, respectively.

We compute an optimal threshold for each element of the state as a decision criterion, given the values of the other elements. Let us consider $\Omega_\alpha$ and $\Omega_\beta$ as the decision criteria at the SA and AA, respectively. To derive the optimal decision $\alpha_n^*$ at the $n$-th slot, there are two alternative ways to define the criterion:
\begin{itemize}
    \item[\textcolor{black}{\rule{0.15cm}{0.15cm}}] $\Omega_\alpha$ is a threshold for the index of the update's importance rank, i.e., $i\in\mathcal{I}$, for $v_n=\nu_i, \forall\nu_i\in\mathcal{V}$, given the E-ACK, i.e., $\hat{E}_n$. Thus, we have
    \begin{equation}
        \alpha_n^* = \mathbbm{1}\big\{i \geq \Omega_\alpha(\hat{E}_n)\,\lvert\,\hat{E}_n\big\}.
    \end{equation}
     \item[\textcolor{black}{\rule{0.15cm}{0.15cm}}] $\Omega_\alpha$ shows a threshold for the E-ACK given the importance rank of the update, such that
    \begin{equation}
        \alpha_n^* = \mathbbm{1}\big\{\hat{E}_n \geq \Omega_\alpha(v_n)\,\lvert\,v_n\big\}.
    \end{equation}
\end{itemize}
For instance in Fig.~\ref{th:fig10a}, $\Omega_\alpha(\hat{E}_n=0)=\Omega_\alpha(\hat{E}_n=1)=7$, $\forall n$. Also, $\Omega_\alpha(v_n=\nu_i)>1$ for $i\leq6$, while $\Omega_\alpha(v_n=\nu_i)=0$ for $i>6$.
Applying the same approach, we find the decision criterion $\Omega_\beta$ to obtain $\beta_n^*$ from three different viewpoints.

\section{Conclusion}
We investigated decision-making for enhancing the effectiveness of updates communicated in the end-to-end status update system based on a push-and-pull communication model. To this end, we considered that the SA observes an information source and generates updates, and its transmission controller decides whether to send them to the AA or not. On the other hand, the AA is responsible for acting based on the updates that are successfully received and the queries initiated to accomplish a subscribed goal at the endpoint. After defining the GoE metric, we formulated the decision problem of finding optimal effect-aware policies that maximize the expected discounted sum of the update’s effectiveness subject to cost constraints. Using the dual problem, we cast it to a CMDP solved separately for each agent based on different model components and proposed an iterative algorithm to obtain the decision policies. Our results established that the push-and-pull model, on average, outperforms the push- and pull-based models in terms of energy efficiency and effectiveness, respectively. Furthermore, effect-aware policies at both agents significantly enhance the effectiveness of updates with a considerable difference in comparison to those of periodic and probabilistic effect-agnostic policies used at either or both agent(s). Finally, we proposed a threshold-based decision policy complemented by a tailored lookup map for each agent that employs effect-aware policies. Future works could explore multi-SA scenarios with varying sensing abilities and realizations' time-dependent importance.

\bibliographystyle{IEEEtran}
\bibliography{References.bib}


\end{document}